\pdfoutput=1
\documentclass[journal,twoside,web]{ieeecolor}



\usepackage{amsthm}
\usepackage{mathrsfs}
\usepackage{lipsum}
\usepackage{generic}

\usepackage{amsmath,amssymb,amsfonts}
\usepackage{graphicx}
\usepackage{algorithm,algorithmic}

\usepackage{textcomp}
\usepackage{xcolor}

\usepackage{setspace}
\usepackage{mathtools}
\usepackage{enumerate}

\usepackage{caption}
\usepackage{subcaption}
\usepackage{algorithm}
\usepackage{float}
\newtheorem{definition}{Definition}
\newtheorem{theorem}{Theorem}

\newtheorem{corollary}{Corollary}

\theoremstyle{definition}

\usepackage{hyperref}

\makeatletter
\let\NAT@parse\undefined
\makeatother
\usepackage{cite}

\newtheorem{innerexample}{Example}[section]

\newcommand{\abs}[1]{\left\lvert#1\right\rvert}

\newcommand{\R}{\mathbb{R}}

\newcommand{\sgn}{\mathrm{sgn}}
\newenvironment{example}[1][]%
{\begin{innerexample}[#1]\pushQED{\qed}}%
{\popQED\end{innerexample}}

\renewenvironment{proof}[1][Proof]{\par\pushQED{\qed}\normalfont\textit{#1.}~}{\popQED\par}

\def\BibTeX{{\rm B\kern-.05em{\sc i\kern-.025em b}\kern-.08em
T\kern-.1667em\lower.7ex\hbox{E}\kern-.125emX}}
\providecommand{\CHG}[1]{#1}

\begin{document}

\title{Nonholonomic Robot Parking by Feedback---\\ Part II: Nonmodular, Inverse Optimal, Adaptive, Prescribed/Fixed-Time or Safe Designs}

\author{Kwang Hak Kim, Velimir Todorovski, and Miroslav Krsti\'{c}
\thanks{This work was supported in part by the Office of Naval Research under Grant No. N00014-23-1-2376, in part by the Air Force Office of Scientific Research under Grant No. FA9550-23-1-0535, and in
part by the National Science Foundation under Grant No. ECCS-2151525. The results and opinions in this paper are solely of the authors and do not reflect the position or the policy of the U.S. Government or the National Science
Foundation. }
\thanks{K. Kim, V. Todorovski, and M. Krstić are with the Department of Mechanical and Aerospace Engineering, UC San Diego, 9500 Gilman Drive, La Jolla, CA, 92093-0411, {\tt\small \{kwk001,vtodorovski,krstic\}@ucsd.edu}}}

\maketitle

\begin{abstract}
For the unicycle system, we provide constructive methods for the design of feedback laws that have one or more of the following properties: being nonmodular and globally exponentially stabilizing, inverse optimal, robust to arbitrary decrease or increase of input coefficients, adaptive, prescribed/fixed-time stabilizing, and safe (ensuring the satisfaction of state constraints). Our nonmodular backstepping feedbacks are implementable with either unidirectional or bidirectional velocity actuation. Thanks to constructing families of strict CLFs for the unicycle, we introduce a general design framework and families of feedback laws for the unicycle, which are inverse optimal with respect to meaningful costs. These inverse optimal feedback laws are endowed with robustness to actuator uncertainty and arbitrarily low input saturation due to the unicycle's driftlessness. Besides ensuring robustness to unknown input coefficients, we also develop adaptive laws for these unknown coefficients, enabling the achievement of good transient performance with lower initial control effort. Additionally, we develop controllers that achieve stabilization within a user-specified time horizon using two systematic methods: time-dilated prescribed-time design with smooth-in-time convergence to zero of both the states and the inputs and homogeneity-based fixed-time control that provides an explicit bound on the settling time. Finally, with a nonovershooting design we guarantee strictly forward motion without curb violation. This article, along with its Part I, lays a broad constructive design foundation for stabilization of the nonholonomic unicycle. 
\end{abstract}


\allowdisplaybreaks

\section{Introduction}

For the unicycle system, 
known from the classic results by Brockett, Ryan, Coron, and Rosier~\cite{brockett1983asymptotic, ryan1994brockett, coron1994relation} 
to not be stabilizable by static state feedback in Cartesian coordinates, in a companion paper~\cite{Part1_todorovskiCLF2025} we introduce a broad framework for designing feedback laws and strict global CLFs for the polar coordinate model. 
Our framework in~\cite{Part1_todorovskiCLF2025} is {\em modular}, meaning that, by separating the design for the feedback laws for the longitudinal and angular velocity inputs, it allows feedback synthesis using the three principal nonlinear feedback design methods---passivity, backstepping, and integrator forwarding. Building on this foundation, in this article we first develop nonmodular stabilizing controllers inspired by~\cite{restrepo2020leader} to achieve exponential stability, and then expand the feedback design methodology for the unicycle to advanced stabilization goals: inverse optimality, adaptive stabilization, prescribed-time stabilization, and constraint-compliant (safe) stabilization.

\vspace{-0.2cm}
\subsection{Related works}

\paragraph{Inverse optimality} Classical optimal control suffers from the curse of dimensionality, making both solving and storing solutions of the Hamilton-Jacobi-Bellman (HJB) equation infeasible for high-dimensional systems, motivating Kalman \cite{kalman1964inverse_opt} to formulate the inverse optimal control problem for linear systems. Introduced for nonlinear systems in \cite{moylan1973nonlinear}, the inverse optimal approach has since been extended in various directions, including stochastic systems \cite{deng1997stochastic} and differential games with bounded disturbances \cite{freeman1996robust}, later generalized to arbitrary disturbances through input-to-state stabilizing controllers \cite{krstic1998inverse}. The first constructive exploration of nonlinear inverse optimal control appears in \cite{sepulchre1997constructive}, for feedback laws of the form $-(L_gV)^{\rm T} = - g^{\rm T}\nabla V $, where $g$ is the input vector field of a system affine in control. This approach has found success in applications such as attitude control of rigid spacecraft \cite{krstic1999inverse,bharadwaj1998geometry}, as well as in inverse optimal safety \cite{krstic2023inverse,krstic2024offender}, yet remains considerably underdeveloped for stabilization of nonholonomic systems. Notably, the work of \cite{do2015global} addresses the stabilization of stochastic nonholonomic systems within the inverse optimal framework; however, establishes inverse optimality only at the subsystem level and without strict CLFs thanks to the inherent properties of stochastic systems. 

\paragraph{Unicycle adaptive control with model uncertainties} Adaptive control for nonholonomic systems is largely limited to time-varying, switching, or hybrid methods for stabilization and tracking problems. Results such as \cite{jiang1995backstepping} achieve global asymptotic stabilization using explicitly time-varying feedback for chained-form systems, while \cite{hespanha1999logic} employs logic-based switching to handle parametric uncertainties. Robust and adaptive switching schemes also provide adaptive control laws that overcome Brockett’s obstruction and address model uncertainties~\cite{jiang2000robust,ge2003adaptive}. Other works, such as \cite{jiangdagger1997tracking}, focus on adaptive tracking rather than regulation. The absence of globally strict CLFs historically prevented the construction of time-invariant adaptive feedback laws for unicycles, resulting in a literature dominated by time-varying or switching schemes. The recent development of globally strict CLFs~\cite{Part1_todorovskiCLF2025,restrepo2020leader} allows for the design of time-invariant adaptive feedback laws.

\paragraph{Prescribed/Fixed-time stabilization} The concepts of fixed-time (FxT) stabilization via nonsmooth Lyapunov inequalities and prescribed-time (PT) stabilization via time-varying feedback are introduced in \cite{polyakov2011nonlinear,sanchez2018class,song2017time,krishnamurthy2020dynamic} for a broad class of nonlinear systems, respectively. For unicycle models, FxT and PT stabilization results remain limited; in the context of nonholonomic source seeking, we present a prescribed-time seeking method in \cite{todorovski2023practical}. In contrast, finite-time convergence is achieved in Dubins vehicles using deadbeat controllers \cite{Krstic2025_Dubins} and in unicycles via sliding-mode methods \cite{thomas2016posture}. In addition, fixed-time algorithms for nonholonomic consensus tracking are developed in \cite{defoort2016fixed}, and for FxT stabilization of nonholonomic systems in chained form in \cite{sanchez2020predefined}.  This paper develops prescribed-time and fixed-time control laws that park the unicycle within a user-specified finite time.

\paragraph{Safety-critical control} Analytical safety filters enable real-time enforcement of safety constraints without relying on reachability analysis or set approximations. Among these, Control Barrier Function (CBF) methods~\cite{wieland2007constructive,ames2016CBF} have become a standard framework for ensuring safety in nonlinear systems, with recent developments in inverse optimal safety filters providing safety with built-in optimality and performance tradeoffs~\cite{krstic2023inverse,lyu2025safetyMarginsISSf}. The simplicity of implementation and rigorous theoretical guarantees enable widespread adoption across robotic platforms, including unicycle models~\cite{lindemann2020control,kim2025robust}. Preceding the development of CBFs, the nonovershooting control technique~\cite{krstic_nonovershooting_2006} offers an alternative analytical route to constraint satisfaction. Unlike optimization-based safety filters, this approach provides feedback laws that inherently enforce strict output constraints without requiring a filter layer.

\subsection{Contributions and Organization}

Building on the modular framework and strict CLFs in the companion paper~\cite{Part1_todorovskiCLF2025}, this work establishes a comprehensive set of \emph{nonmodular}, \emph{inverse optimal}, \emph{adaptive}, \emph{prescribed-time}, \emph{fixed-time}, and \emph{safety-critical} stabilization results for the unicycle system. The main contributions are:

\paragraph{Nonmodular stabilization with strict CLFs} In Section~\ref{sec:nonmodular} we develop nonmodular stabilization feedback laws, following~\cite{restrepo2020leader}, which achieve global exponential stability with backstepping for both unidirectional and bidirectional velocity.

\paragraph{General inverse optimal control for the unicycle} In Section~\ref{sec:IOC} we exploit our construction of globally strict CLFs in~\cite{Part1_todorovskiCLF2025} to design an extensive family of inverse optimal controllers, including linear, sublinear, bounded, and ``relay-approximating'' feedback laws, providing a systematic way to trade off control effort and convergence behavior while preserving global stability. Portions of the results presented without proofs appear in~\cite{kim2025inverseoptimalfeedbackgain}.

\paragraph{Gain margin and robustness} We show that the inverse optimal controllers inherit an \emph{infinite gain margin}, a consequence of the unicycle’s driftless structure. This yields robustness to actuator uncertainty and allows stabilization even under arbitrarily low input saturation levels.

\paragraph{Adaptive control with unknown input coefficients} 
In Section~\ref{sec:adapt}, we design adaptive feedback and update laws that compensate for unknown input coefficients while preserving global asymptotic stability. Our approach leverages the globally strict CLFs in~\cite{Part1_todorovskiCLF2025}, which enable smooth and time-invariant adaptive feedback design.
Unlike prior adaptive stabilization results, our design avoids the trajectory zigzagging due to oscillating gains inherent to time-varying methods~\cite{jiang1995backstepping}, does not rely on the assumption of known bounds on the uncertain parameters as in~\cite{jiang2000robust}, and eliminates the potential chattering or excessive oscillations from switching control laws in~\cite{hespanha1999logic,ge2003adaptive}.
Adaptation allows the initial gains to be low while the state is large, and the gains larger when the state becomes small---enabling good transients at reduced (early) peaks in control effort.

\paragraph{Prescribed/Fixed-time stabilization} In Section~\ref{sec:PT}, by properly scaling the nonmodular stabilizing control laws, we guarantee convergence to the origin in a user defined prescribed/fixed-time, with inputs that are not only bounded but whose values, and derivatives, also go to zero in prescribed/fixed-time. The scaling factor in the developed PT-control laws is time-varying and smooth in time and it is derived based on the temporal transformation approach. In contrast, similar to the use of homogeneous scaling factors in time-varying control design \cite{murray1997homogeneous}, where  asymptotic stability is achieved, but only with exponential (infinite-time) convergence, we introduce a homogeneous state-dependent scaling factor that enforces a suitable Lyapunov inequality, thereby guaranteeing fixed-time convergence of the closed-loop system. Consequently, the proposed designs provide the first control laws that guarantee user-prescribed finite-time convergence while exhibiting a less restrictive region of attraction than existing approaches.
\paragraph{Safety design} In Section~\ref{sec:safety}, we add curb-like position constraints on top of the stabilization objective with strictly forward motion, guaranteeing inherent safety without performance-degrading safety filters.

\section{Unicycle in Polar Coordinates}\label{sec:unicycle}

Consider the standard unicycle model given as
\begin{subequations}
\label{eq:unicycle_cartesian}
\begin{eqnarray}
\dot{x} &=&  v \cos(\theta)\\
\dot{y} &=& v \sin(\theta)  \\
\dot{\theta} &=& \omega \,,\label{eq:dot_theta_omega}
\end{eqnarray}
\end{subequations}
where $(x,y) \in \R^2$ is the position of the unicycle in Cartesian coordinates, $\theta \in \mathbb{R}$ is the heading angle, $v$ is the forward velocity input, and $\omega$ is the angular velocity input. 
We introduce the polar transformation of the system \eqref{eq:unicycle_cartesian} as shown in Table~\ref{tab:polar_coordinates}, obtaining  the unicycle model in polar coordinates:
\begin{subequations}
\label{eq:unicycle_polar}
\begin{eqnarray}
\label{eq:unicycle_polar_rhodot}
\dot{\rho} &=& 
-v \cos\gamma\\
\label{eq:unicycle_polar_deltadot}
\dot{\delta} &=&  \frac{v}{\rho} \sin\gamma
\\
\label{eq:unicycle_polar_gammadot}
\dot{\gamma} &=&  \frac{v}{\rho} \sin\gamma -\omega  \,.
\end{eqnarray} 
\end{subequations}

\begin{figure}[t!]
\centering
\includegraphics[width=.6\linewidth]{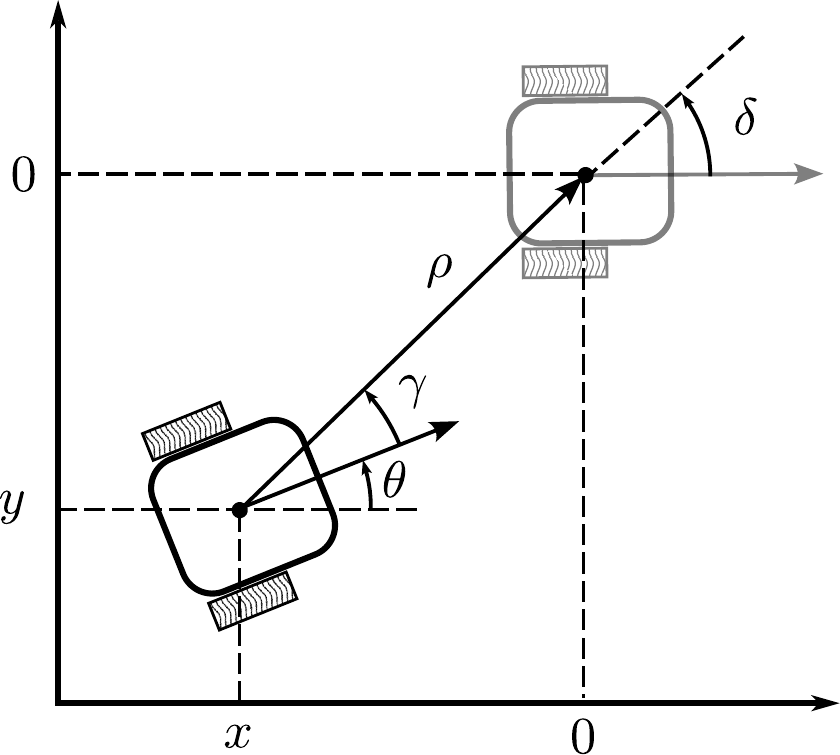}
\caption{Unicycle orientation $(x, y, \theta)$ relative to the goal state $(0, 0, 0)$, and the corresponding polar coordinate transformation $(\rho, \gamma, \delta)$.}
\vspace{-0.4cm}
\label{fig:unicycle_cord}
\end{figure}

\paragraph{State spaces and stability notions}
We define the following state-spaces:
\begin{align}
\mathcal{S}_0   &:= \left\{ \rho > 0 \right\} \times \mathcal{T}_0,
& \quad \mathcal{T}_0   &:= \left\{ \delta \in \mathbb{R},\, \gamma \in \mathbb{R} \right\} \label{eq:ss_S}\\[2pt]
\mathcal{S}_1 &:= \left\{ \rho > 0 \right\} \times \mathcal{T}_1,
& \quad \mathcal{T}_1 &:= \left\{ \delta \in \mathbb{R},\, \abs{\gamma} < \pi \right\} \label{eq:ss_S1}\\[2pt]
\mathcal{S}_2 &:= \left\{ \rho > 0 \right\} \times \mathcal{T}_2,
& \quad \mathcal{T}_2 &:= \left\{ \abs{\delta} < \pi,\, \gamma \in \mathbb{R} \right\} \label{eq:ss_S2}\\[2pt]
\mathcal{S}_3 &:= \left\{ \rho > 0 \right\} \times \mathcal{T}_3,
& \quad \mathcal{T}_3 &:= \left\{ \abs{\delta} < \pi,\, \abs{\gamma} < \pi \right\} \label{eq:ss_S3}\,,
\end{align}
the functions:
\begin{align*}
\Delta_0(s)=\Delta_1(s)=\Gamma_0(s)=\Gamma_2(s)&:=
s\\
\Delta_2(s)=\Delta_3(s)=\Gamma_1(s)=\Gamma_3(s)&:=
2\tan\left(\frac{s}{2}\right)\,,
\end{align*}
and equip the state spaces~\eqref{eq:ss_S}--\eqref{eq:ss_S3} with the metrics
\begin{equation*}
\left|(\rho,\delta,\gamma)\right|_{\mathcal{S}_i}
:=
\rho
+
\left|(\delta,\gamma)\right|_{\mathcal{T}_i}
:=
\rho
+
\left|\Delta_i(\delta)\right|
+
\left|\Gamma_i(\gamma)\right| \,,
\end{equation*}
where $i\in\{0,1,2,3\}$. The following definitions are adapted from~\cite[Defs.~1 and~2]{Part1_todorovskiCLF2025}, where a detailed discussion is given.
\begin{definition}
\label{def-our-GAS}
Consider system~\eqref{eq:unicycle_polar} in closed loop with a feedback pair $(v,\omega)$ that is continuous on the state space $\hat{\mathcal{S}} \in \{\mathcal{S}_0,\mathcal{S}_1,\mathcal{S}_2,\mathcal{S}_3\}$ with respect to its metric. If there exists a class $\mathcal{KL}$ function $\beta$ such that, for all $t\geq 0$, it holds that $|(\rho(t),\delta(t), \gamma(t))|_{\hat{\mathcal{S}}}\leq \beta\left( |(\rho_0,\delta_0, \gamma_0)|_{\hat{\mathcal{S}}},t \right)$, we say that the point $\rho=\delta = \gamma=0$ is {\em globally asymptotically stable on $\hat{\mathcal{S}}$} (GAS on $\hat{\mathcal{S}}$).
\end{definition}
\begin{definition}[CLF for the unicycle \eqref{eq:unicycle_polar}]
\label{def-CLF}
A continuously differentiable function $(\rho,\delta,\gamma)\mapsto V$
is a \textit{control Lyapunov function} (CLF) with respect to \eqref{eq:unicycle_polar} if (i) there exist class $\mathcal{K}$ functions  $(\bar\alpha_1,\bar\alpha_2)$ such that,  for all $(\rho,\delta,\gamma)$ in
$\Sigma = \{\rho\geq 0\}\times \hat{\mathcal{T}}$, where $\hat{\mathcal{T}}\in\{\mathcal{T}_0,\mathcal{T}_1, \mathcal{T}_2, \mathcal{T}_3\}$,
$ \bar\alpha_1(|(\rho,\delta, \gamma)|_{\hat{\mathcal{S}}}) \le V(\rho,\delta,\gamma) \le \bar\alpha_2(|(\rho,\delta, \gamma)|_{\hat{\mathcal{S}}})$, where $\hat{\mathcal{S}} = \{\rho> 0\}\times \hat{\mathcal{T}}$, and (ii) there exists $(v/\rho,\omega)\in \mathbb{R}^2$ such that for all $(\rho,\delta,\gamma)\neq (0,0,0)$ in $\Sigma$
\begin{equation*}
\left[-\dfrac{\partial V}{\partial\rho}\rho\cos\gamma + \left(\dfrac{\partial V}{\partial\delta}+\dfrac{\partial V}{\partial\gamma}\right)\sin\gamma \right]\dfrac{v}{\rho} - \dfrac{\partial V}{\partial\gamma}\omega < 0 \,.
\end{equation*}
\end{definition}

Following Definitions~\ref{def-our-GAS} and~\ref{def-CLF}, the state spaces \eqref{eq:ss_S}--\eqref{eq:ss_S3} delimit the domain of each design: the bound $\abs{\delta} < \pi$ (spaces $\mathcal{S}_2,\mathcal{S}_3$) excludes positions directly in front of the target and $\abs{\gamma} < \pi$ (spaces $\mathcal{S}_1,\mathcal{S}_3$) excludes headings pointing exactly away from it, with the barrier CLFs rendering the chosen space forward-invariant. Hence, every claim of \emph{global} stability below holds \emph{on} the stated $\mathcal{S}_i$, i.e., for all initial conditions in $\mathcal{S}_i$.

\paragraph{Roadmap}
The designs of this article are redesigns built on two shared foundations: the strict CLFs of~\cite{Part1_todorovskiCLF2025} and the nonmodular backstepping laws of Section~\ref{sec:nonmodular}. Each section redesigns the feedback law with the following: inverse optimality by constructive feedback from the CLFs (Section~\ref{sec:IOC}), unknown input coefficients adaptation (Section~\ref{sec:adapt}), user-assigned convergence time by rescaling the feedback (Section~\ref{sec:PT}), and safe parking via the barrier CLF of~\cite[Thm.~8]{Part1_todorovskiCLF2025} (Section~\ref{sec:safety}). Each design in this article is posed on the state space among \eqref{eq:ss_S}--\eqref{eq:ss_S3} matched to its objective, as summarized in Table~\ref{tab:thm_state_spaces}.


\begin{table}[t!]
\centering
\begin{tabular}{|l|l|l|}
\hline
$\rho = \sqrt{x^2 + 
y^2}$  & Distance to origin \\
\hline
 $\delta = \text{{\rm atan2}}(y ,x )  + \pi$
& Polar  angle
\\
\hline
$\gamma = \text{{\rm atan2}}(y ,x ) - \theta  + \pi$
& Line-of-sight angle   \\
\hline
\end{tabular}
\caption{Polar coordinates and their expressions in terms of Cartesian coordinates. 
The transformation $(x,y,\theta)\mapsto (\rho,\delta,\gamma)$ is discontinuous on $x<0,y=0$ and not defined at $x=y=0$. If the target is at $(x^*, y^*, \theta^*)\neq 0$, the transformation generalizes to $\rho=\sqrt{(x-x^*)^2+(y-y^*)^2}, \delta= \text{{\rm atan2}}(y-y^* ,x-x^* ) -\theta^* + \pi, \gamma= \delta-\theta+\theta^*$.
}
\vspace{-0.15cm}
\label{tab:polar_coordinates}
\end{table}

\begin{table}[t]
\centering
\begin{tabular}{lcl}
\hline
Result & State space & Property \\
\hline
Theorem~\ref{thm:unidirectional_bkstp} & $\mathcal{S}_2$ & GES, unidirectional $v$ \\
Theorem~\ref{thm:bidirectional_bkstp} & $\mathcal{S}_0$ & GES, bidirectional $v$ \\
Theorem~\ref{thrm:IOC} & $\mathcal{S}_0,\mathcal{S}_1,\mathcal{S}_2,\mathcal{S}_3$ & inverse optimality \\
Theorem~\ref{thrm:adapt} & $\mathcal{S}_0,\mathcal{S}_1,\mathcal{S}_2,\mathcal{S}_3$ & adaptive regulation \\
Theorem~\ref{thm:PT-control} & $\mathcal{S}_0$ or $\mathcal{S}_2$ & PT stabilization \\
Theorem~\ref{thm:FxT-control} & $\mathcal{S}_0$ or $\mathcal{S}_2$ & FxT stabilization \\
Theorem~\ref{thrm:nonovershoot} & $\mathcal{S}_2$ & GAS, nonovershooting \\
\hline
\end{tabular}
\caption{State spaces on which the results of this paper hold.}
\vspace{-0.4cm}
\label{tab:thm_state_spaces}
\end{table}

\section{Nonmodular Unicycle Stabilization}\label{sec:nonmodular}

\label{sec:nonmodular_unicycle_stabilization}
The companion paper \cite{Part1_todorovskiCLF2025} develops a {\em modular} framework for stabilizing controllers and global strict CLFs for the unicycle. Modularity allows multiple design methods (passivity, integrator forwarding, backstepping) but guarantees only GAS, not global exponential stability (GES). Notably, sacrificing modularity enables GES via a backstepping design~\cite{restrepo2020leader} that jointly determines the forward velocity and steering inputs in the first step for the $(\rho,\delta)$-subsystem. This, however, comes at the cost of higher design complexity and the lack of extension beyond backstepping to passivity or forwarding methods. Despite these limitations, we develop two backstepping extensions of the nonmodular approach in \cite{restrepo2020leader}. Relative to \cite[Sec.~III.A]{restrepo2020leader}, which achieves GES with unidirectional velocity on $\mathcal{S}_0$, Theorem~\ref{thm:unidirectional_bkstp} poses the unidirectional design on $\mathcal{S}_2$, whose barrier prevents the vehicle from crossing in front of the target, while Theorem~\ref{thm:bidirectional_bkstp} attains GES on $\mathcal{S}_0$ with bidirectional velocity, as compared in Fig.~\ref{fig:loria_vs_thm9}.

\subsection{Unidirectional GES BAR-FLi}

In many practical systems, such as fixed-wing aircraft or guided missiles, reversing the forward velocity is infeasible. Motivated by this, we first present an extension for the GES controllers from \cite{restrepo2020leader} to prevent crossing in front of the target under the constraint of unidirectional forward velocity. 

\begin{figure}[t]
\centering
\includegraphics[width=.73\linewidth]{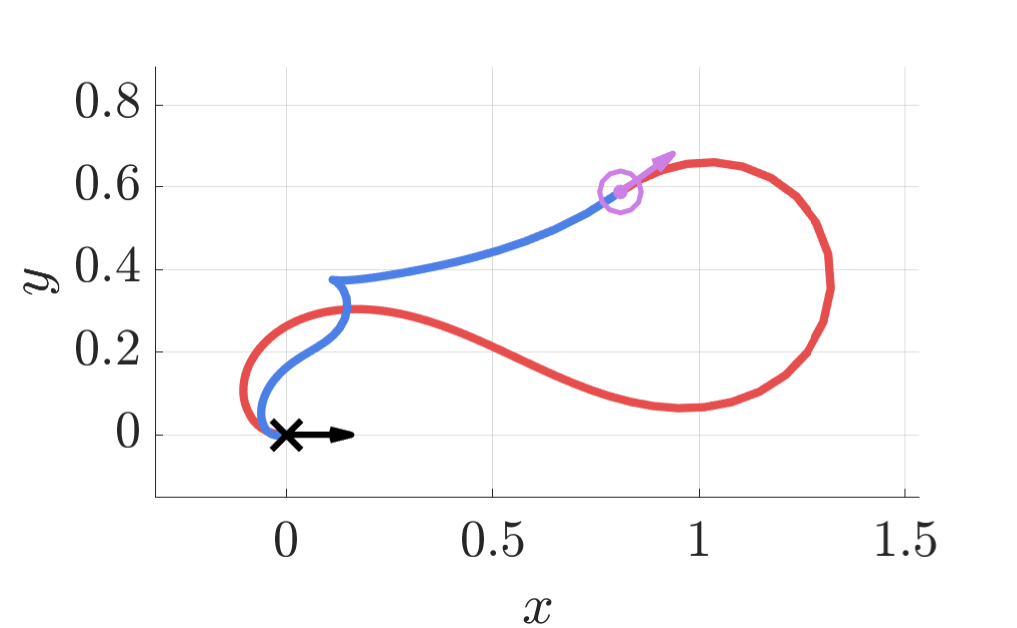}
\caption{Comparison of the controllers of Theorem~\ref{thm:bidirectional_bkstp} (blue) and \cite[Sec.~III.A]{restrepo2020leader} (red) from $(\rho_0,\delta_0,\gamma_0) = (1,-4\pi/5,\pi)$. Restricting to unidirectional forward velocity yields a less efficient trajectory to the target.}
\label{fig:loria_vs_thm9}
\vspace{-0.7cm}
\end{figure}

\begin{theorem}[Unidirectional BAR-FLi]
\label{thm:unidirectional_bkstp}
For the 
system \eqref{eq:unicycle_polar}, 
the feedback laws
\begin{eqnarray}\label{eq:unidir_control_general}
v &=& k_1\sigma(\delta)\rho
\label{eq:unidir_control_v}
\\
\omega &=& \dfrac{v}{\rho} \sin(\gamma) +\tilde\omega
\label{eq:unidir_control_omega}\\
\tilde\omega &=& k_4z - k_3\rho^2\sigma(\delta)\frac{\partial\psi(z,\gamma)}{\partial \gamma} + k_3\frac{V_\delta(\delta)}{\phi(\delta)}\sigma(\delta)\psi(z,\gamma)\nonumber\\
&& +\frac{k_1k_2}{\sigma^2(\delta)}\phi'(\delta)\left(\sigma(\delta)\psi(z,\gamma)z - \phi(\delta)\right)\label{eq:unidir_control_omegabar}\,,
\end{eqnarray}
where
\begin{eqnarray}
V_\delta(\delta) &=& 4\tan^2\frac{\delta}{2}
\\
\phi(\delta) &=& \sin\delta\\
\sigma(\delta) &=& \sqrt{1+ k_2^2\phi^2(\delta)}\\
z &=& \gamma + \arctan\left(k_2\phi(\delta)\right)\label{eq:unidir_z}\\
\psi(r,s) &=& \frac{\sin(r-s) +\sin(s)}{r}\label{eq:GES_psi}\,,
\end{eqnarray}
and $k_1,k_2,k_3,k_4 >0$, renders the point ${\rho = \delta = \gamma = 0}$ GES on $\mathcal{S}_2$, i.e., there exist $K \geq 1$ and $\lambda > 0$ such that, for all $t\geq 0$, it holds that $|(\rho(t),\delta(t), \gamma(t))|_{\mathcal{S}_2}\leq K|(\rho_0,\delta_0, \gamma_0)|_{\mathcal{S}_2}e^{-\lambda (t-t_0)}$.
Furthermore, 
\begin{align}\label{eq:unidir_V_general}
V = \rho^2 + V_\delta + q^2\left[\gamma + \arctan\left(k_2\phi(\delta)\right)\right]^2\,,
\end{align}
with $q = \sqrt{k_1/k_3}$ 
is a globally strict CLF for \eqref{eq:unicycle_polar} on $\mathcal{S}_2$ with respect to the input pair $(v/\rho,\omega)$ in the sense of~Definition~\ref{def-CLF}. 
\end{theorem}

\begin{proof}
First, we note that $\sigma(\delta)\cos(z-\gamma) = 1$ and $\sigma(\delta)\sin(z-\gamma) = k_2\phi(\delta)$. With the backstepping transformation \eqref{eq:unidir_z} and substituting \eqref{eq:unidir_control_v} into \eqref{eq:unicycle_polar_rhodot} and \eqref{eq:unicycle_polar_deltadot} yields
\begin{eqnarray}
\dot\rho &=& -k_1\rho - k_1\rho\sigma(\delta)\frac{\partial\psi(z,\gamma)}{\partial\gamma}z \label{eq:rho_dot_unidir}\\
\dot\delta &=& -k_1k_2\phi(\delta) + k_1\sigma(\delta)\psi(z,\gamma)z\,. \label{eq:delta_dot_unidir}
\end{eqnarray}
Moreover, with \eqref{eq:unidir_control_omega}, we get
\begin{equation}
\dot z = \frac{k_1k_2}{\sigma^2(\delta)}\phi'(\delta)\left(\sigma(\delta)\psi(z,\gamma)z - k_2\phi(\delta)\right) - \tilde\omega\,. \label{eq:z_dot_unidir}
\end{equation}
Then, taking the time derivative of \eqref{eq:unidir_V_general} gives
\begin{align}
\dot V =& -2k_1\rho^2  -2k_1k_2V_\delta(\delta) \nonumber\\
&+2q^2z\biggl[-k_3\rho^2\sigma(\delta)\frac{\partial\psi(z,\gamma)}{\partial \gamma}+ k_3\frac{V_\delta(\delta)}{\phi(\delta)}\sigma(\delta)\psi(z,\gamma)\nonumber\\
&+ \frac{k_1k_2}{\sigma^2(\delta)}\phi'(\delta)\left(\sigma(\delta)\psi(z,\gamma)z - k_2\phi(\delta)\right) - \tilde\omega\biggr]\,, 
\end{align}
which, with \eqref{eq:unidir_control_omegabar}, yields
\begin{equation}\label{eq:unidir_Vdot1}
\dot V = -2k_1\rho^2  - 2k_1k_2V_\delta(\delta) - 2k_4q^2z^2 \le - \underline{c} V \,,
\end{equation}
 where $\underline c = \min\{2k_1,2k_1k_2, 2k_4\}>0$. Hence, \eqref{eq:unidir_V_general} is a strict CLF for \eqref{eq:unicycle_polar}, and by the comparison principle $|(\rho(t),\delta(t),\gamma(t))|_{\mathcal{S}_2} \leq K |(\rho_0,\delta_0,\gamma_0)|_{\mathcal{S}_2} e^{-\lambda (t-t_0)}$, where $K(k_2) > 1$ and $\lambda = \underline{c}/2$.
\end{proof}

\subsection{Bidirectional GES Backstepping}

When reversing the forward velocity is permitted, exploiting this capability is more efficient, especially in parking scenarios. The following theorem does so and achieves GES. Fig.~\ref{fig:loria_vs_thm9} compares the resulting trajectory with that of \cite[Sec.~III.A]{restrepo2020leader}.


\begin{theorem}[Bidirectional Backstepping]
\label{thm:bidirectional_bkstp}
For the 
system \eqref{eq:unicycle_polar},  the feedback laws
\begin{eqnarray}\label{eq:bidir_control_general}
v &=& k_1\rho\sigma(\delta)\cos\gamma
\label{eq:bidir_control_v}
\\
\omega &=& \dfrac{v}{\rho} \sin \gamma +\tilde\omega
\label{eq:bidir_control_omega}\\
\tilde\omega &=& k_4z - k_3\rho^2\sigma(\delta)\psi_2(2z,2\gamma)+k_3\delta\sigma(\delta)\psi(2z,2\gamma) \nonumber\\
&& +\frac{k_1k_2}{\sigma^2(\delta)}\left(\sigma(\delta)\psi(2z,2\gamma)z - k_2\delta\right)\label{eq:bidir_control_omegabar}\,,
\end{eqnarray}
where
\begin{eqnarray}
\sigma(\delta) &=& \sqrt{1+ (2k_2\delta)^2}\\
z &=& \gamma + \frac{1}
{2}\arctan\left(2k_2\delta\right)\label{eq:bidir_z}\\
\psi_2(r,s) &=& \frac{\partial \psi(r,s)}{\partial s}\,,
\end{eqnarray}
with $\psi(r,s)$ defined in \eqref{eq:GES_psi} and $k_1,k_2,k_3,k_4 >0$, renders the point $\rho = \delta = \gamma = 0$ GES on $\mathcal{S}_0$, i.e., there exist $K \geq 1$ and $\lambda > 0$ such that, for all $t\geq 0$, it holds that $|(\rho(t),\delta(t), \gamma(t))|_{\mathcal{S}_0}\leq K|(\rho_0,\delta_0, \gamma_0)|_{\mathcal{S}_0}e^{-\lambda (t-t_0)}$.
Furthermore, 
\begin{align}\label{eq:bidir_V_general}
V = \rho^2 + \delta^2 + q^2\left[\gamma + \frac{1}{2}\arctan\left(2k_2\delta\right)\right]^2\,,
\end{align}
with $q = \sqrt{k_1/k_3}$ 
is a globally strict CLF for \eqref{eq:unicycle_polar} on $\mathcal{S}_0$ with respect to the input pair $(v/\rho,\omega)$ in the sense of~Definition~\ref{def-CLF}. 
\end{theorem}

\begin{proof}
With the backstepping transformation \eqref{eq:bidir_z} and substituting \eqref{eq:bidir_control_v} into \eqref{eq:unicycle_polar_rhodot} and \eqref{eq:unicycle_polar_deltadot} yields
\begin{eqnarray}
\dot\rho &=& -\frac{1}{2}k_1\rho(1+\sigma(\delta)) -k_1\rho\sigma(\delta)\psi_2(2z,2\gamma)z \label{eq:rho_dot_bidir}\\
\dot\delta &=& -k_1k_2\delta + k_1\sigma(\delta)\psi(2z,2\gamma)z\,. \label{eq:delta_dot_bidir}
\end{eqnarray}
Moreover, with \eqref{eq:bidir_control_omega}, we get
\begin{equation}
\dot z = \frac{k_1k_2}{\sigma^2(\delta)}(\sigma(\delta)\psi(2z,2\gamma)z - k_2\delta) - \tilde\omega\,. \label{eq:z_dot_bidir}
\end{equation}
Then, taking the time derivative of \eqref{eq:bidir_V_general} gives
\begin{align}
\dot V =& -k_1\rho^2\left(1+\sigma(\delta)\right) -2k_1k_2\delta^2 \nonumber\\
&+2q^2z\biggl[-k_3\rho^2\sigma(\delta)\psi_2(2z,2\gamma) + k_3\delta\sigma(\delta)\psi(2z,2\gamma)\nonumber\\
&+ \frac{k_1k_2}{\sigma^2(\delta)}\left(\sigma(\delta)\psi(2z,2\gamma)z - k_2\delta\right) - \tilde\omega\biggr]\,, 
\end{align}
which, with \eqref{eq:bidir_control_omegabar}, yields
\begin{equation}\label{eq:bidir_Vdot1}
\dot V = -k_1\rho^2(1+\sigma(\delta))  - 2k_1k_2\delta^2 - 2k_4q^2z^2 \le -\underline{c} V \,,
\end{equation}
 where $\underline c = \min\{k_1,2k_1k_2, 2k_4\}>0$. Hence, \eqref{eq:bidir_V_general} is a strict CLF for \eqref{eq:unicycle_polar}, and by the comparison principle $|(\rho(t),\delta(t),\gamma(t))|_{\mathcal{S}_0} \leq K |(\rho_0,\delta_0,\gamma_0)|_{\mathcal{S}_0} e^{-\lambda (t-t_0)}$, where ${K(k_2) > 1}$ and ${\lambda = \underline{c}/2}$.
\end{proof}

\section{Inverse Optimal Design}\label{sec:IOC}

We now turn to the design of optimal controllers. Even for the unicycle, whose state dimension is only three, the HJB PDE of classical optimal control rarely admits a solution in closed form. The appeal of inverse optimality is that it delivers a meaningful optimality guarantee without solving the HJB PDE. In this section, we derive a family of inverse optimal controllers for the unicycle based on the numerous CLFs available. It is important to note that the derived controllers fall into the category of damping $L_gV$ controllers 
(i.e., see \cite{jurdjevic1978controllability},  \cite[Prop. 5.9.1]{sontag2013mathematical}) that immediately satisfy the Jurdjevic-Quinn conditions and globally stabilize the unicycle.

\subsection[Alternative controllers: LgV (gradient) designs]{Alternative controllers: $L_gV$ (gradient) designs}
\label{sec-alternatives}

We return to the unicycle model 
\eqref{eq:unicycle_polar}.
If we take the $L_gV$ approach in its vanilla form, we run into a problem: the input vector field associated with the forward velocity input $v$ in the polar model has a singularity at $\rho=0$. We circumvent this singularity in a way that is not just mathematically resourceful but physically meaningful. 
We introduce a change of velocity control variable from $v$ to $v/\rho$,
where $v/\rho$ is a new control input to be designed, and where the velocity feedback $v$ is forced to vanish at the target location, $\rho=0$. It is to be expected that the velocity be proportional to the distance from the target---the vehicle needs to slow down in order to stop at the target. 

Having replaced the input pair $(v,\omega)$ by $(v/\rho,\omega)$, we now have a new input vector field $\bar g = [\rho g_1, g_2]$, with which $L_{\bar g} V$ is ``desingularized.'' From \eqref{eq:unicycle_polar}, the vector fields are given by
\begin{align}
\quad g_1 &\coloneqq\dfrac{1}{\rho}\bar g_1\coloneqq \dfrac{1}{\rho}{\begin{bmatrix}
-\rho\cos\gamma \\
\sin\gamma \\
\sin\gamma 
\end{bmatrix}} \,, \ \  g_2 := {\begin{bmatrix}
0 \\
0 \\
-1
\end{bmatrix}}\,.
\end{align}
And we are now ready to study the design of controllers that employ $L_{\bar g} V$, as well as their stabilizing and inverse optimality properties. 

As we shall see shortly, the inverse optimal design will rely on the property that 
\begin{center}
$L_{\bar g} V(\Xi) = 0$ only at $\Xi=0$, where $\Xi := [\rho,\delta,\gamma]^\top$,    
\end{center}
which in turn implies that
\begin{center}
$|L_{\bar g} V(\Xi)|$ is positive definite in $\Xi$.
\end{center}
This property holds for the following reason. For systems affine in control, $\dot x = f(x) + g(x)u$, a positive definite radially unbounded differentiable function $V(x)$ is a CLF if and only if the following implication holds: for all $x\neq 0$ where $L_gV(x)=0$ it holds that $L_fV(x) < 0$. When $f(x)\equiv 0$, this implication gives that $L_gV(x)\neq 0$ for all $x\neq 0$, namely, that $L_gV(x)=0 $ only at $x=0$. For instance, we establish that all Lyapunov functions in~\cite[Thm. 1]{Part1_todorovskiCLF2025} are CLFs in the sense of~Definition~\ref{def-CLF} with the Genova feedback law~\cite{aicardi1995}. The feedback law of~\cite{aicardi1995} achieves asymptotic parking through a nonstrict Lyapunov analysis. The strict CLFs of~\cite{Part1_todorovskiCLF2025} certify this closest classical design with explicit class-$\mathcal{KL}$ estimates, and every redesign of this paper hinges on their strictness, beginning with the inverse optimality redesign of this section. 

While our theorems establish that we construct Lyapunov functions $V(\rho,\delta,\gamma)$ such that $|L_{\bar g} V(\Xi)|$ is positive definite in $\Xi$, it is not radially unbounded. Radial unboundedness can be achieved with a modification of the CLF, which is described in the final segment of the proof of~\cite[Thm. 3.2]{krstic1998inverse}.

As an illustration of what the Lie derivatives $-L_{\bar{g}_1}V$ and $-L_{g_2}V$ look like, we compute them for the Genova CLF in~\CHG{\cite[Eq. (30) with (31)]{Part1_todorovskiCLF2025}},
with unity gains:
\begin{subequations}
\label{eq:LgV_Genova-both}
\begin{align}
-L_{\bar{g}_1}V/2 &= \rho^2\cos\gamma - \sin\gamma\left(\delta^2 + \gamma^2 + 4\right)(\delta+\gamma)\label{eq:Lg1V_Genova}\\
-L_{g_2}V/2 &= \left(\delta^2 + \gamma^2 + 2\right)\gamma + (\delta+\gamma)
\,.
\end{align}
\end{subequations}
Similarly, the Lie derivatives $-L_{\bar{g}_1}V$ and $-L_{g_2}V$ for the GloBa CLF in~\CHG{\cite[Eq. (109) with (112)]{Part1_todorovskiCLF2025}},
with unity gains, yield
\begin{subequations}
\label{eq:LgV_GloBa-both}
\begin{align}
-L_{\bar{g}_1}V/2 &= \rho^2\cos\gamma  - \sin\gamma\left[\delta + \left(1+\frac{1}{1+4\delta^2}\right)z\right] \label{eq:Lg1V_GloBa}\\
-L_{g_2}V/2 &= z
\,.
\label{eq:Lg2V_GloBa}
\end{align}
\end{subequations}
where $z = \gamma + \frac{1}{2}\arctan(2\delta)$. In both cases, the introduction of the change of control variable from $v$ to $v/\rho$ 
has successfully ``desingularized'' the Lie derivative expressions as seen in \eqref{eq:Lg1V_Genova} and \eqref{eq:Lg1V_GloBa}. 

\subsection{Basic quadratic inverse optimality}
\label{sec-basic-quadratic}

Before presenting the general inverse optimality methodology in Section \ref{sec-gen-inv-opt}, we give a particular example of what we are after. For any of our previously designed CLFs $V(\rho,\delta,\gamma)$, all controllers of the form 
\begin{subequations}\label{eq-u*-quad0}  
\begin{align}
\label{eq-LQ-forward}
v^* &= -\rho\varepsilon_1^2L_{\bar{g}_1}V\nonumber\\
&:=  \varepsilon_1^2\rho\left[\dfrac{\partial V}{\partial\rho}\rho\cos\gamma - \left(\dfrac{\partial V}{\partial\delta}+\dfrac{\partial V}{\partial\gamma}\right) \sin\gamma \right]\\
\label{eq-LQ-steer}
\omega^* &= -\varepsilon_2^2L_{g_2}V := \varepsilon_2^2\dfrac{\partial V}{\partial\gamma}\,,
\end{align}
\end{subequations}
for all $\varepsilon_1, \varepsilon_2>0$, are not only globally asymptotically stabilizing but are the minimizers of the parametrized costs
\begin{equation}
J = \int_0^\infty \Biggl[l(\rho,\delta,\gamma) +\left(\dfrac{v}{\varepsilon_1\rho}\right)^2 +\left(\dfrac{\omega}{\varepsilon_2}\right)^2  \Biggr]{\rm d}t\,,
\end{equation}
where 
\begin{equation}
l(\rho,\delta,\gamma) = \left(\varepsilon_1 L_{\bar{g}_1}V\right)^2 +\left(\varepsilon_2L_{g_2}V\right)^2\,.
\end{equation}
Such costs are meaningful in the sense of imposing (i.) a zero penalty when all the state variables $\rho,\delta,\gamma$ and both of the ($\rho$-weighted) control inputs $v/\rho$ and $\omega$ are zero, and (ii.) a positive penalty when any of the states and either of the (weighted) controls are nonzero. Note that when $\rho = 0$, the forward velocity input is subjected to an infinite penalty, which is desirable, as it discourages  movement away from the positional origin once it is reached.

It is relevant to examine the complexity of the $L_gV$ controllers as compared to their counterparts used to derive the CLFs. To that end, we focus on the GloBa controller and note that the inverse optimal steering feedback \eqref{eq-LQ-steer} with \eqref{eq:Lg2V_GloBa} is much simpler than the non-optimal GloBa feedback~\CHG{\cite[Eq. (7) with (108)]{Part1_todorovskiCLF2025}}. But, in contrast, compared to the basic forward velocity feedback $v=\rho\cos\gamma$, the inverse optimal forward velocity feedback \eqref{eq-LQ-forward} with \eqref{eq:Lg1V_GloBa}, given by $v^* = \rho^3\cos\gamma  - \rho\sin\gamma\left[\delta + \left(1+\frac{1}{1+4\delta^2}\right)z\right]$ is considerably more complex and depends not only on $(\rho,\gamma)$ but also on $\delta$ (when the vehicle points sidewise relative to the target, the velocity is nonzero). 

\subsection{General inverse optimal designs}
\label{sec-gen-inv-opt}

To generalize the inverse optimality result from the end of Section \ref{sec-basic-quadratic}, we  recall the \textit{Legendre–Fenchel transform}, defined as follows.

\begin{definition}(Legendre-Fenchel transform)\label{def:legendre_Fenchel}. Let $\eta$ be a class $\mathcal{K}_\infty[0,a)$ function whose derivative $\eta'$ is also a class $\mathcal{K}_\infty[0,a)$ function where $a > 0$ is finite or infinite.
The mapping 
\begin{align}
\ell \eta(r) = \int_0^r (\eta')^{-1}(s) ds
\end{align}
represents the Legendre-Fenchel transform, where $(\eta')^{-1}(r)$ stands for the inverse function of $d\eta(r)/dr$.
\end{definition} 

The transformation has useful properties, given in~\cite[Lemma A1]{krstic1998inverse}. Also, to minimize notational repetition, we use the subscript $i$ in place of the subscripts $1$ and $2$, corresponding to the vector fields $\bar g_1$ and $g_2$ and their associated control inputs $v/\rho$ and $\omega$.

\begin{theorem} \label{thrm:IOC} Consider the system \eqref{eq:unicycle_polar} rewritten as
\begin{align}\label{eq:ioc_sys}
\dot{\Xi} = \bar{g}_1(\Xi)\frac{v}{\rho} + g_2(\Xi)\omega
\end{align}
where $\Xi \coloneqq [\rho,\delta,\gamma]^\top$ and
\begin{align}
\bar{g}_1 = \begin{bmatrix}
-\rho\cos\gamma\\
\sin\gamma\\
\sin\gamma
\end{bmatrix}, \quad
g_2 = \begin{bmatrix}
0\\
0\\
-1
\end{bmatrix}\,.
\end{align}
For any $\eta_i \in \mathcal{K}_\infty[0,a_i)$ such that also $\eta_i' \in \mathcal{K}_\infty[0,a_i)$ where $a_i > 0$ is finite or infinite,
and for any continuous positive scalar-valued functions $\varepsilon_1(\rho, \delta, \gamma)$ and $\varepsilon_2(\rho, \delta, \gamma)$, the cost functional
\begin{equation}\label{eq:ioc_J_thrm}
J = \int_0^\infty \Biggl[l(\rho,\delta,\gamma) + \eta_1\left(\frac{|v|}{\varepsilon_1\rho}\right) + \eta_2\left(\frac{|\omega|}{\varepsilon_2}\right)\Biggr] {\rm d}t\,,
\end{equation}
where
\begin{equation}
l(\rho,\delta,\gamma) = \ell\eta_1(\varepsilon_1|\nu_1|) + \ell\eta_2(\varepsilon_2|\nu_2|)\,,
\end{equation} is minimized by the feedback law
\begin{subequations}
\label{eq:ioc_u*}
\begin{align}
v^* &= -\rho\varepsilon_1 (\eta_1')^{-1}(\varepsilon_1|\nu_1|)\sgn(\nu_1)\\
\omega^* &= -\varepsilon_2 (\eta_2')^{-1}(\varepsilon_2|\nu_2|)\sgn(\nu_2)\,,
\end{align}
\end{subequations}
with $\nu_1(\Xi) \coloneqq L_{\bar{g}_1}V$ and $\nu_2(\Xi) \coloneqq L_{g_2}V$ denoted for the CLFs given by \cite[Thms. 1-8]{Part1_todorovskiCLF2025},
and  
for all initial conditions on the respective state spaces $\mathcal{S}_0,\mathcal{S}_1,\mathcal{S}_2,$ and $\mathcal{S}_3$. Additionally, the feedback law
\begin{subequations}
\label{eq:ioc_u}
\begin{align}
v &= -\rho\varepsilon_1 \frac{\ell\eta_1(\varepsilon_1|\nu_1|)}{\varepsilon_1|\nu_1|}\sgn(\nu_1)\\
\omega &= -\varepsilon_2\frac{\ell\eta_2(\varepsilon_2|\nu_2|)}{\varepsilon_2|\nu_2|}\sgn(\nu_2)\,,
\end{align}
\end{subequations}
is continuous in $(\nu_1,\nu_2)$ and renders the point $\rho=\delta = \gamma = 0$ of the system \eqref{eq:ioc_sys} GAS on the respective state spaces $\mathcal{S}_0,\mathcal{S}_1,\mathcal{S}_2,$ and $\mathcal{S}_3$, in accordance to~Definition~\ref{def-our-GAS}.
\end{theorem}

\begin{proof}
\textbf{Step 1:} \eqref{eq:ioc_u} is continuous in $(\nu_1,\nu_2)$ and stabilizes \eqref{eq:ioc_sys}. First, from Definition~\ref{def:legendre_Fenchel} and applying the Leibniz rule yields $ (\ell\eta)'(r) = (\eta')^{-1}(r)$. Then, applying L'Hôpital's rule, we get 
\begin{align}
\label{eq-L'Hopital}
\lim_{r\rightarrow 0}\frac{\ell\eta(r)}{r} = \lim_{r\rightarrow 0}\frac{(\eta')^{-1}(r)}{1}\,.
\end{align}
Since $\eta'(r)$ is a class $\mathcal{K}_\infty$ function, its inverse $(\eta')^{-1}(r)$ is also a class $\mathcal{K}_\infty$ function, which by definition is zero at $r = 0$. Thus, the limit in \eqref{eq-L'Hopital} is  zero  and the feedback law \eqref{eq:ioc_u} is continuous in $(\nu_1,\nu_2)$.

Secondly, plugging in \eqref{eq:ioc_u} into $\dot{V}$ yields
\begin{align}
\dot{V}|_{\eqref{eq:ioc_u}} &= \nu_1 \frac{v}{\rho} + \nu_2\omega \nonumber\\
&= -\ell\eta_1(\varepsilon_1|\nu_1|)  - \ell\eta_2(\varepsilon_2|\nu_2|)< 0\label{eq:ioc_V_stable}\,.
\end{align}
By \cite[Lemma A1]{krstic1998inverse}, $\ell\eta_i$ are class $\mathcal{K}_\infty$ functions, and since $V$ is a strict CLF, it follows by definition that $\nu_1(\Xi) = 0$ and $\nu_2(\Xi) = 0$ if and only if $\Xi = 0$. Hence, \eqref{eq:ioc_V_stable} holds, and the point $\rho = \delta =\gamma = 0$ is GAS under each CLF on its respective state space when using the feedback law \eqref{eq:ioc_u}.

\textbf{Step 2:} \eqref{eq:ioc_u*} stabilizes \eqref{eq:ioc_sys}. Plugging in \eqref{eq:ioc_u*} into $\dot{V}$ yields
\begin{align}
\dot{V}|_{\eqref{eq:ioc_u*}} &= \nu_1 \frac{v^*}{\rho} + \nu_2\omega^* \nonumber\\
&= -\varepsilon_1|\nu_1|(\eta_1')^{-1}(\varepsilon_1|\nu_1|)-\varepsilon_2|\nu_2|(\eta_2')^{-1}(\varepsilon_2|\nu_2|)\label{eq:ioc_V_stable*_1}\,.
\end{align}
From \cite[Lemma A1]{krstic1998inverse}, we see that $r(\eta')^{-1}(r) = \ell \eta(r) + \eta((\eta')^{-1}(r))$ and \eqref{eq:ioc_V_stable*_1} can be rewritten as
\begin{align}
\dot{V}|_{\eqref{eq:ioc_u*}} &= -\sum_i \ell\eta_i(\varepsilon_i|\nu_i|) -\sum_i\eta_i((\eta_i')^{-1}(\varepsilon_i|\nu_i|))\nonumber\\
&\leq -\ell\eta_1(\varepsilon_1|\nu_1|)  - \ell\eta_2(\varepsilon_2|\nu_2|)  = \dot{V}|_{\eqref{eq:ioc_u}}< 0.
\end{align}
Thus, the point $\rho = \delta = \gamma = 0$ is GAS under each CLF on its respective state space when using the feedback law \eqref{eq:ioc_u*}.

\textbf{Step 3:} \eqref{eq:ioc_u*} minimizes \eqref{eq:ioc_J_thrm}. First, we rearrange \eqref{eq:ioc_J_thrm} and consider the fact that $\lim_{t\rightarrow\infty}V(\Xi(t)) = 0$, yielding
\begin{align}
J &= \int_0^\infty \biggl[\ell\eta_1(\varepsilon_1|\nu_1|) + \ell\eta_2(\varepsilon_2|\nu_2|)\nonumber\\
&\qquad \qquad \qquad+ \eta_1\left(\frac{|v|}{\varepsilon_1\rho}\right) + \eta_2\left(\frac{|\omega|}{\varepsilon_2}\right) + \dot{V} - \dot{V}\biggr]{\rm d}t\nonumber\\
&= V(\Xi(0)) - \lim_{t\rightarrow\infty}V(\Xi(t)) \nonumber\\
&\qquad + \int_0^\infty \biggl[\ell\eta_1(\varepsilon_1|\nu_1|) + \ell\eta_2(\varepsilon_2|\nu_2|) \nonumber\\
&\qquad \qquad \qquad + \eta_1\left(\frac{|v|}{\varepsilon_1\rho}\right) + \eta_2\left(\frac{|\omega|}{\varepsilon_2}\right) + \nu_1\frac{v}{\rho} + \nu_2\omega\biggr]{\rm d}t\nonumber\\
&= V(\Xi(0)) +\int_0^\infty \biggl[\ell\eta_1(\varepsilon_1|\nu_1|) + \eta_1\left(\frac{|v|}{\varepsilon_1\rho}\right) + \nu_1\frac{v}{\rho}\nonumber\\
&\qquad \quad \qquad +\ell\eta_2(\varepsilon_2|\nu_2|) + \eta_2\left(\frac{|\omega|}{\varepsilon_2}\right) +  \nu_2\omega\biggr] \label{eq:ioc_J_proof_1}{\rm d}t\,.
\end{align}
Using  \cite[Thm. 156]{hardy_inequalities_1989} and the fact that $\rho > 0$ for all $\mathcal{S}_0,\mathcal{S}_1,\mathcal{S}_2,$ and $\mathcal{S}_3$, we get the following upper bounds
\begin{align}
-\nu_1\frac{v}{\rho} &= \left(\frac{v}{\varepsilon_1\rho }\right)\left(-\varepsilon_1\nu_1\right) \leq \eta_1\left(\frac{|v|}{\varepsilon_1\rho}\right) + \ell\eta_1(\varepsilon_1|\nu_1|)\\
-\nu_2\omega &= \left(\frac{\omega}{\varepsilon_2}\right)\left(-\varepsilon_2\nu_2\right)\leq \eta_2\left(\frac{|\omega|}{\varepsilon_2}\right) + \ell\eta_2(\varepsilon_2|\nu_2|)\,.
\end{align}
Then, multiplying both sides with $-1$ yields
\begin{align}
\nu_1\frac{v}{\rho} &\geq -\eta_1\left(\frac{|v|}{\varepsilon_1\rho}\right) - \ell\eta_1(\varepsilon_1|\nu_1|)\label{eq:ioc_young1}\\
\nu_2\omega &\geq -\eta_2\left(\frac{|\omega|}{\varepsilon_2}\right)- \ell\eta_2(\varepsilon_2|\nu_2|)\label{eq:ioc_young2}\,,
\end{align}
and we observe that the cost functional \eqref{eq:ioc_J_proof_1} achieves a minimum when \eqref{eq:ioc_young1} and \eqref{eq:ioc_young2} are equalities. The result from \cite[Thm. 156]{hardy_inequalities_1989} states that \eqref{eq:ioc_young1} and \eqref{eq:ioc_young2} are equalities if and only if
\begin{align}
x_1 = (\eta_1')^{-1}(|y_1|)\frac{y_1}{|y_1|}, \;\; \text{and}\;\; x_2 = (\eta_2')^{-1}(|y_2|)\frac{y_2}{|y_2|},
\end{align}
where $x_1 \coloneqq \frac{v}{\varepsilon_1\rho},\; y_1 \coloneqq-\varepsilon_1\nu_1, \; x_2 \coloneqq \frac{\omega}{\varepsilon_2},$ and $y_2 \coloneqq-\varepsilon_2\nu_2$, which implies
\begin{align}
\frac{v}{\varepsilon_1\rho} &= -(\eta_1')^{-1}(\varepsilon_1|\nu_1|)\frac{\varepsilon_1\nu_1}{\varepsilon_1|\nu_1|}\nonumber\\
&\implies v = -\rho\varepsilon_1 (\eta_1')^{-1}(\varepsilon_1|\nu_1|)\sgn(\nu_1) = v^*\\
\frac{\omega}{\varepsilon_2} &= -(\eta_2')^{-1}(\varepsilon_2|\nu_2|)\frac{\varepsilon_2\nu_2}{\varepsilon_2|\nu_2|}\nonumber\\
&\implies \omega = -\varepsilon_2 (\eta_2')^{-1}(\varepsilon_2|\nu_2|)\sgn(\nu_2) = \omega^*\,.
\end{align}
Thus, 
the inequalities \eqref{eq:ioc_young1} and \eqref{eq:ioc_young2} are equalities if and only if $v = v^*$ and $\omega = \omega^*$, which yields
\begin{subequations}\label{eq:ioc_young_equality}
\begin{align}
\nu_1\frac{v}{\rho} &= -\eta_1\left(\frac{|v|}{\varepsilon_1\rho}\right) - \ell\eta_1(\varepsilon_1|\nu_1|)\\
\nu_2\omega &= -\eta_2\left(\frac{|\omega|}{\varepsilon_2}\right)- \ell\eta_2(\varepsilon_2|\nu_2|)\,.
\end{align}
\end{subequations}
Substituting \eqref{eq:ioc_young_equality} into \eqref{eq:ioc_J_proof_1} yields $J = V(\Xi(0))$. Thus, \eqref{eq:ioc_u*} minimizes \eqref{eq:ioc_J_thrm} for each CLF on its respective state space.
\end{proof}

The designs of this section extend in principle to other driftless nonholonomic systems, but only once strict CLFs are constructed for the system at hand, as accomplished for the unicycle in \cite{Part1_todorovskiCLF2025}.

While the $L_gV$ controllers achieve optimality, they do not retain the local exponential stabilizing property of the basic stabilizing controllers, showcased in \CHG{Sections V-E, VI-C and VII-C} of \cite{Part1_todorovskiCLF2025}. This is evident, for example, from the $L_{\bar g_1}V$ expressions \eqref{eq:Lg1V_Genova} and \eqref{eq:Lg1V_GloBa}, which are locally quadratic in the state near the origin, and therefore cannot be exponentially stabilizing. Such a linearity-lacking dependence stems from the $L_gV$ design methodology itself, whose feedback is proportional to $L_{\bar g_1}V$ and $L_{g_2}V$. 

\subsection{Specific inverse optimal choices}

\begin{table}[t]
\renewcommand{\arraystretch}{2.0}
\centering
\small
\begin{tabular}{|c|c|c|}
\hline
$\eta(r)$ & $(\eta')^{-1}(r)$ & $\ell\eta(r)/r$ \\[4pt]
\hline \hline
$r^2/2$ & $r$ & $r/2$ \\[4pt]
\hline
$\cosh(r) - 1$ & $\operatorname{arcsinh}(r)$ & $\operatorname{arcsinh}(r) - \dfrac{r}{\sqrt{r^2 + 1} + 1}$ \\[10pt]
\hline
$-\ln(\cos(r))$ & $\arctan(r)$ & $\arctan(r) - \dfrac{\ln(1 + r^2)}{2r}$ \\[10pt]
\hline
$\dfrac{e}{e^{{1}/{r}} - e}$ & $\dfrac{1}{1 + \ln\left(1 + 1/r \right)}$ & 
$ \dfrac{1}{r} \displaystyle\int  \dfrac{\text{d} r}{1 + \ln\left(1 + 1/r \right)}
$\\[10pt]
\hline
\end{tabular}
\caption{The cost-on-control functions $\eta(\cdot)$ and the resulting expressions that define  the inverse optimal and stabilizing $L_gV$ controllers.}
\label{tab:cost-on-control-functions}
\vspace{-0.4cm}
\end{table}

\begin{figure*}[t]
\centering
\begin{subfigure}[t]{0.24\textwidth}
\centering
\includegraphics[width=\textwidth]{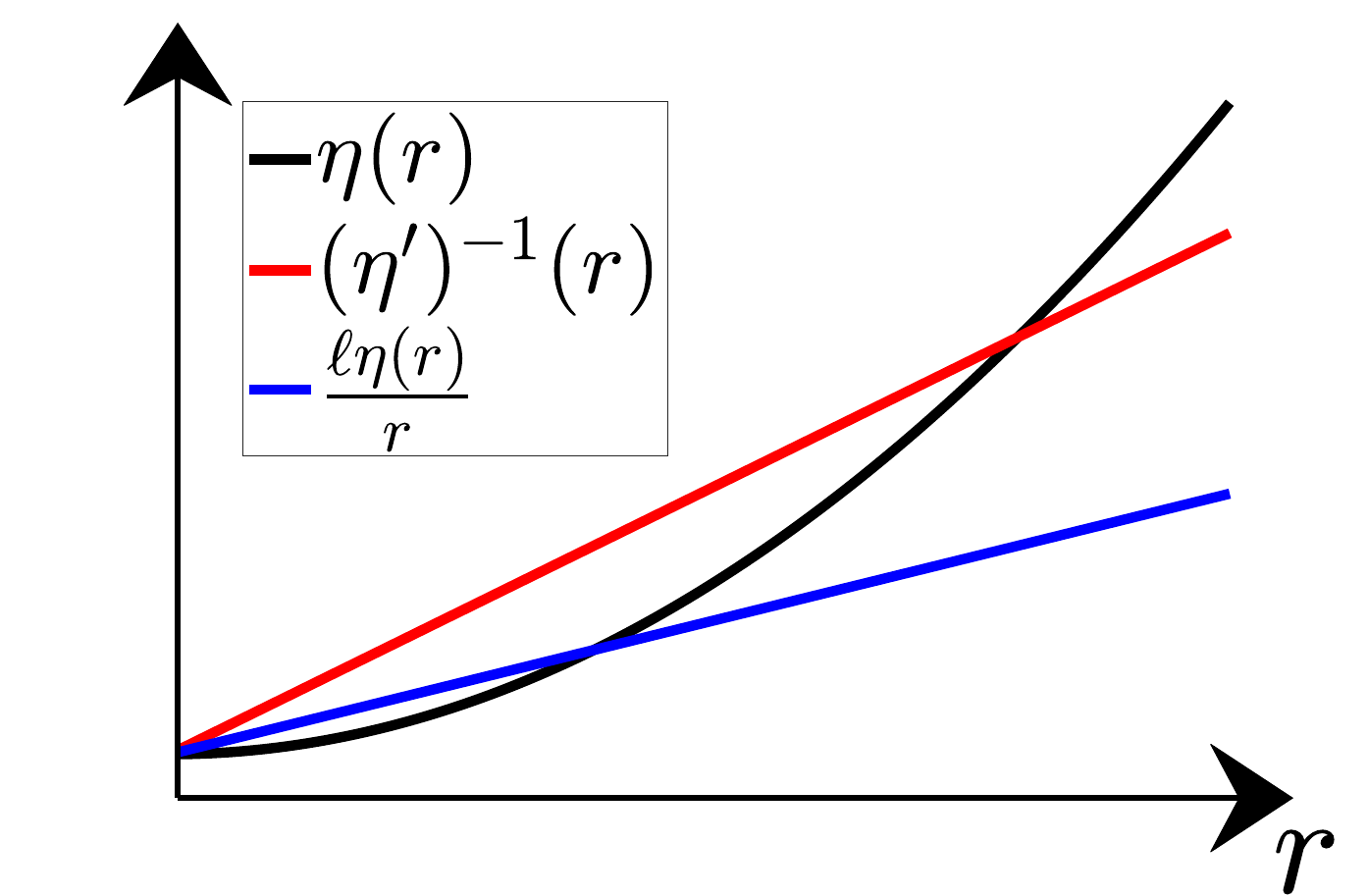}
\caption{$\eta(r) = r^2 / 2$ in (black), $(\eta^{\prime})^{-1}(r) = r$ in (red) and $\frac{\ell \eta (r)}{r} = r/2$ in (blue).}
\label{fig:eta_1}
\end{subfigure}
\begin{subfigure}[t]{0.24\textwidth}
\centering
\includegraphics[width=\textwidth]{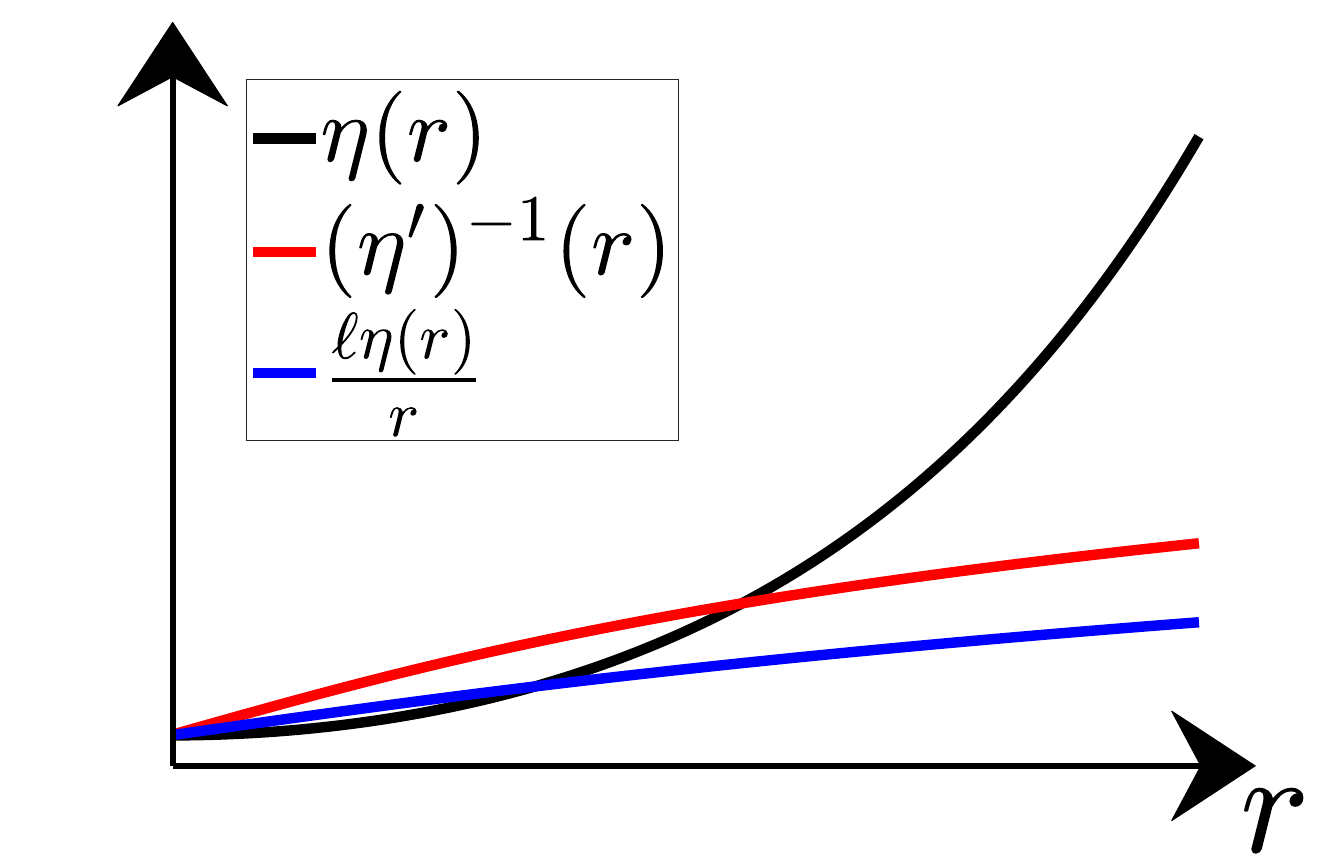}
\caption{$\eta(r) = \cosh(r) - 1$ in (black), $(\eta^{\prime})^{-1}(r) = {\rm arcsinh}(r)$ in (red) and $\frac{\ell \eta (r)}{r} = {\rm arcsinh}(r) - r/(1+\sqrt{r^2+1})$ in (blue).}
\label{fig:eta_2}
\end{subfigure}
\begin{subfigure}[t]{0.24\textwidth}
\centering
\includegraphics[width=\textwidth]{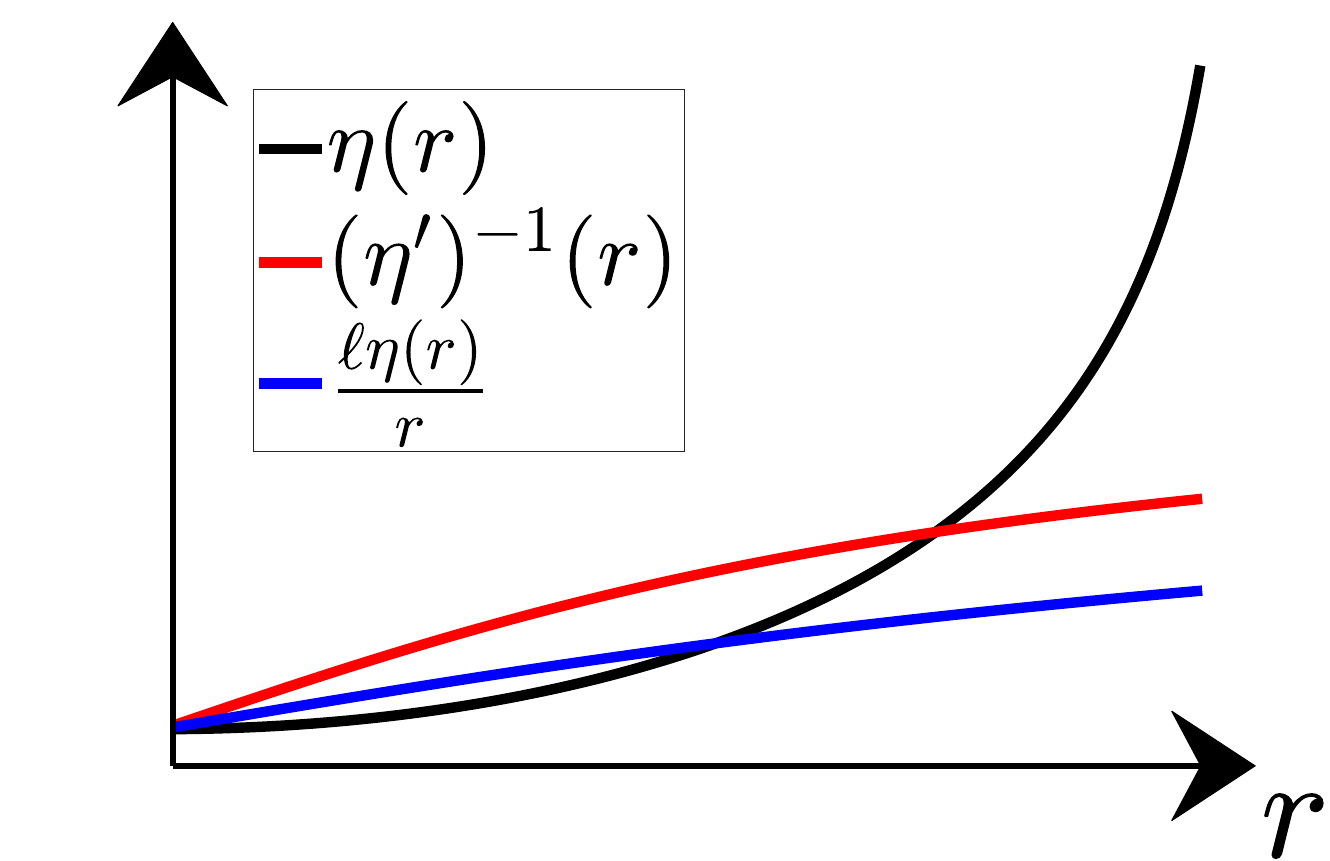}
\caption{$\eta(r) = -\ln(\cos(r))$ in (black), $(\eta^{\prime})^{-1}(r) = \arctan(r)$ in (red) and $\frac{\ell \eta (r)}{r} = \arctan(r) - \ln(1+r^2)/2r$ in (blue).}
\label{fig:eta_3}
\end{subfigure}
\begin{subfigure}[t]{0.24\textwidth}
\centering
\includegraphics[width=\textwidth]{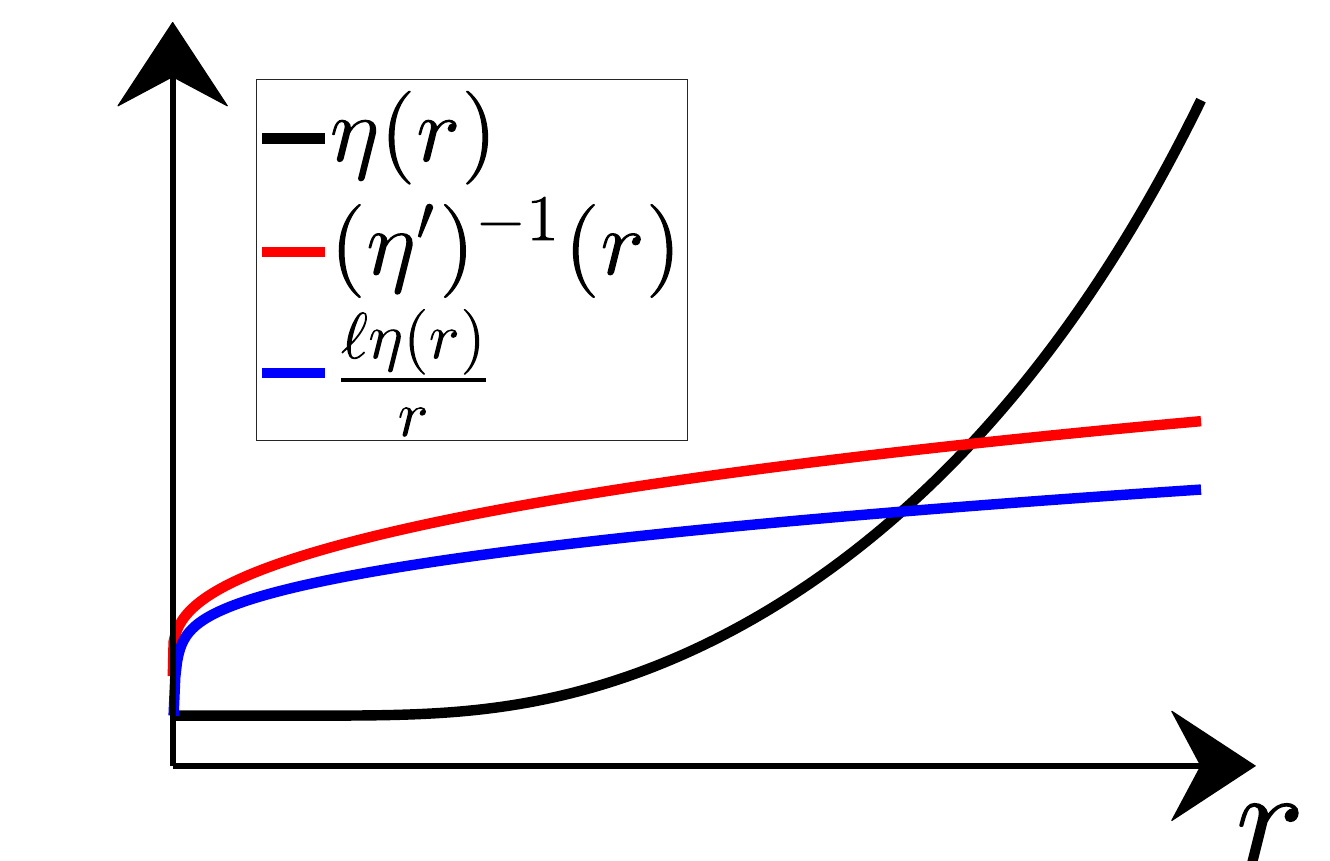}
\caption{$\eta(r) = e/(e^{1/r} - e)$ in (black), $(\eta^{\prime})^{-1}(r) = 1/(1+\ln(1+1/r))$ in (red) and $\frac{\ell \eta (r)}{r} = \frac{1}{r} \int \frac{{\rm d}r}{1+\ln(1+1/r)}$ in (blue).}
\vspace{-0.4cm}
\label{fig:eta_4}
\end{subfigure}

\caption{Cost-on-control functions from Table \ref{tab:cost-on-control-functions}.}
\label{fig:cost-on-control functions}
\end{figure*}

The design space of the stabilizer \eqref{eq:ioc_u} and its optimal counterpart \eqref{eq:ioc_u*} includes many parameters and infinitely many CLF choices. An exhaustive exploration is impossible, but one can explore qualitatively the tradeoff between the control effort and the related cost on the state.

We present four examples based on different choices of the function $\eta$, each leading to a distinct inverse optimal controller. Example~\ref{example:IOC1} yields a controller linear in $L_{\bar{g}_1}V$ and $L_{g_2}V$, already introduced in Section \ref{sec-basic-quadratic}, Example~\ref{example:IOC2} produces a sublinear relation of control with  $L_{\bar{g}_1}V$ and $L_{g_2}V$, and Examples~\ref{example:IOC3} and \ref{example:IOC4} result in bounded control laws. While Theorem \ref{thrm:IOC} does not require uniformity in the choice of $\eta_1$ and $\eta_2$, the following examples use the same function for both to highlight each behavior, and  the subscript on $\eta(r)$ is omitted henceforth. The four choices of $\eta$ are summarized in Table~\ref{tab:cost-on-control-functions}, while Figure~\ref{fig:cost-on-control functions} illustrates how the functions increase with respect to the input, reflecting the corresponding rise in control effort cost.

Additionally, since inverse optimal redesign relies on $L_gV$-based controllers, different CLFs naturally lead to different system behaviors. Hence, we conduct simulation comparisons between the examples using a particular GloBa CLF in the following corollary derived from \cite[Thm. 7]{Part1_todorovskiCLF2025}:
\begin{corollary}\label{cor:com_CLF}
For the system \eqref{eq:unicycle_polar}, the Lyapunov function
\begin{equation}\label{eq:comp_CLF}
V(\rho,\delta,\gamma) = \sqrt{1 + k_1\rho^2} + \sqrt{1+V_{\delta \gamma}(\delta,\gamma)} - 2\,,
\end{equation}
where
\begin{equation}
V_{\delta \gamma}(\delta,\gamma) = \delta^2 + k_3 \left(\gamma + \frac{1}{2}\arctan(2k_2 \delta) \right)^2\,,
\end{equation}
with $k_1,k_2,k_3 >0$ is a globally strict CLF on $\mathcal{S}_0$ for \eqref{eq:unicycle_polar} with respect to the input pair $(v/\rho, \omega)$ in the sense of Definition~\ref{def-CLF}.
\end{corollary}

Compared with \cite[Thm.~7]{Part1_todorovskiCLF2025}, the gain $q^2 = k_1/k_3$ is rescaled as $k_3$, which makes the tuning of the $L_gV$ control laws more direct while preserving the result of~\cite[Thm.~7]{Part1_todorovskiCLF2025}.

\begin{example}\label{example:IOC1} \textbf{(Linear in $L_{\bar{g}_1}V$ and $L_{g_2}V$)}
Consider the quadratic cost function defined by
\begin{align}\label{eq:eta_quad}
\eta(r) = \frac{r^2}{2}
\,,
\end{align} 
with the derivative inverse $(\eta')^{-1}(r) = r$ and the Legendre-Fenchel transform of $\ell\eta = \frac{r^2}{2}$. With this choice $\eta$, one finds the classical quadratic cost, where the running cost is quadratic in $\nu_1$ and $\nu_2$ and the penalty is quadratic in the control, 
as detailed in the basic quadratic inverse optimal controller \eqref{eq-u*-quad0} in Section~\ref{sec-basic-quadratic}, with control linear in $L_{\bar{g}_1}V$ and $L_{g_2}V$.

The simulation results for this controller with $\varepsilon_1 = \varepsilon_2 = 1$, shown in Fig.~\ref{fig:ioc_trajectory}, has a clear contrast to the analytically derived controllers in \cite{Part1_todorovskiCLF2025} (see Fig. 4 in \cite{Part1_todorovskiCLF2025}). Notably, for many initial conditions, the controller chooses to reverse at first, adjusting its position to achieve a more favorable alignment before moving forward toward the target. This differs from the analytical controller, which tended to favor minimal deviation from the shortest path. Furthermore, for the initial condition where the unicycle starts parallel and directly above the target (cyan), the behavior resembles a conventional parallel parking maneuver (albeit reversed), more aligned with how a human driver would intuitively approach the task.

\begin{figure}[t]
\centering
\includegraphics[width=\linewidth]{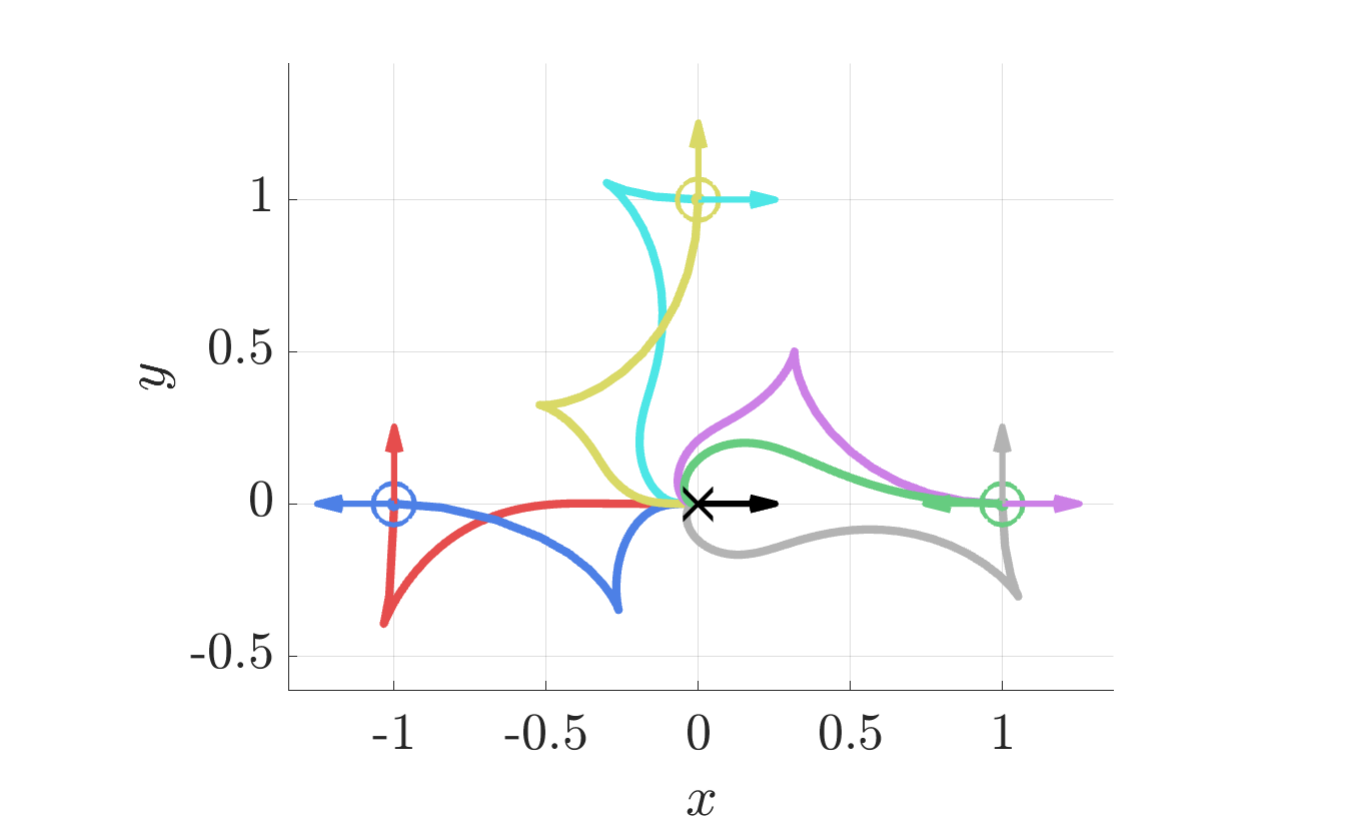}
\caption{Trajectory of the inverse optimal controller designed using the CLF~\eqref{eq:comp_CLF} with gains $(k_1,k_2,k_3) = (6.5,3,7)$ and the quadratic control cost~\eqref{eq:eta_quad}. The trajectory differences between Examples~\ref{example:IOC1}–\ref{example:IOC4} are negligible and therefore omitted. Different colors correspond to different initial conditions.}
\vspace{-0.4cm}
\label{fig:ioc_trajectory}
\end{figure}

However, quadratic costs on control and state may be naive for nonlinear systems in general, particularly for the unicycle. Although they penalize input over the time horizon, they can still yield occasional large values outside actuator limits. As shown in Fig.~\ref{fig:ioc_controleffort_ex1}, this results in forward velocity magnitudes up to about $10$ units/s and steering rate magnitudes up to about $2.5$ units/s---values not unreasonable but potentially beyond actuator range. In the next three examples, we present alternatives with stronger control penalties. Since the trajectories $y(x)$ remain nearly identical across cases, we focus instead on the differences in control effort.
\end{example}


\begin{figure}[t]
\centering
\begin{subfigure}{0.5\textwidth}
\centering
\includegraphics[width=.9\textwidth]{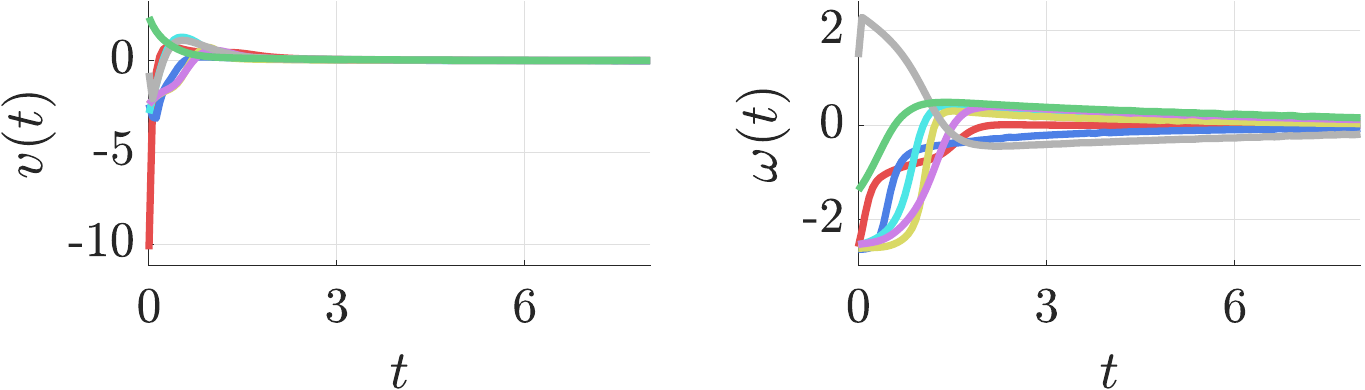}
\caption{Control effort of the linearly proportional to $L_{\bar{g}_1}V$ and $L_{g_2}V$ controller \eqref{eq-u*-quad0}.}
\label{fig:ioc_controleffort_ex1}
\end{subfigure}
\vspace{0.25cm}
\begin{subfigure}{0.5\textwidth}
\centering
\includegraphics[width=.9\textwidth]{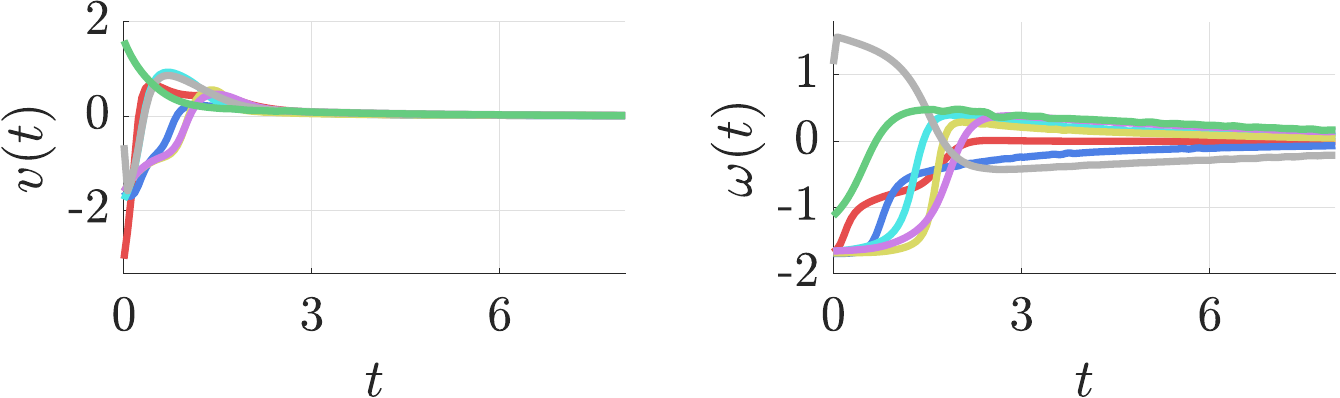}
\caption{Control effort of the sublinearly proportional to $L_{\bar{g}_1}V$ and $L_{g_2}V$ controller \eqref{eq-u*-cosh}.}
\label{fig:ioc_controleffort_ex2}
\end{subfigure}
\caption{Comparison in control effort between Example~\ref{example:IOC1} and Example~\ref{example:IOC2} with $\varepsilon_1 = \varepsilon_2 = 1$. Each color denotes the same corresponding initial conditions as in Fig.~\ref{fig:ioc_trajectory}.}
\vspace{-0.4cm}
\label{fig:ioc_ex1_vs_ex2}
\end{figure}

\begin{example}\label{example:IOC2} \textbf{(Sublinear/logarithmic in $L_{\bar{g}_1}V$ and $L_{g_2}V$)}
Consider the hyperbolic cosine cost function defined as
\begin{align}\label{eq:eta_hyperbolic}
\eta(r) = \cosh(r) - 1
\,,
\end{align}
with the associated derivative inverse $(\eta')^{-1}(r) = \mathrm{arcsinh}(r)$ and the Legendre-Fenchel transformation $\ell\eta(r) =  r \,\mathrm{arcsinh}(r)+ 1 - \sqrt{1+r^2}$. 
The  cost $\eta(r) = \cosh(r) - 1$ imposes a larger penalty when the control effort is large as compared to the quadratic cost in Example~\ref{example:IOC1}. This is reflected in the minimizer feedback
\begin{subequations}\label{eq-u*-cosh}
\begin{align}
v^* &= -\rho\varepsilon_1\mathrm{arcsinh}\left(\varepsilon_1L_{\bar{g}_1}V\right)\\
\omega^* &= -\varepsilon_2\mathrm{arcsinh}\left(\varepsilon_2L_{g_2}V\right)\,,
\end{align}
\end{subequations}
where compared to the linear optimal controller in \eqref{eq-u*-quad0}, the feedback law in \eqref{eq-u*-cosh} exhibits a sublinear dependence on $L_{\bar{g}_1}V$ and $L_{g_2}V$. Specifically, as $q = \varepsilon_iL_{g_i}V$ grows in magnitude, we observe that
\begin{align}
{\mathrm{arcsinh}(q)} &= {\ln\left(\sqrt{1+|q|^2}+ |q| \right)}\sgn(q)\nonumber\\
&\approx {\ln(1+|q|)}\sgn(q)\,,
\end{align}
indicating that the feedback scales logarithmically with $L_{\bar{g}_1}V$ and $L_{g_2}V$. This implies that the controller grows more slowly than linearly with respect to $L_{\bar{g}_1}V$ and $L_{g_2}V$, and thus applies less aggressive control effort. The feedback in \eqref{eq-u*-cosh} can be interpreted as a less effort-intensive alternative to the linear controller in \eqref{eq-u*-quad0}. 
This is reflected in Fig.~\ref{fig:ioc_controleffort_ex2}, with $\varepsilon_1 = \varepsilon_2 = 1$, showing a nearly threefold reduction in control effort in the forward velocity input.
\end{example}

\begin{example}\label{example:IOC3}\textbf{(Bounded/saturating control)} While the hyperbolic cost on control effort significantly reduces the input magnitudes, one might still desire a stronger guarantee by directly enforcing a bound on the maximum allowable control input. 
To achieve bounded optimal control, we require that the inverse of $\eta'$ in the general expression for the optimal control be a bounded function. For instance, setting $(\eta')^{-1}(r) = \arctan(r)$ leads to
\begin{equation}\label{eq:eta_ln}
\eta(r) = - \ln(\cos(r))
\,,
\end{equation}
which is of class $\mathcal{K}_\infty$ on the interval $[0,\pi/2)$. 

This choice implies that the control cost becomes unbounded as $|v|$ approaches the limit $\rho\varepsilon_1\pi/2$ and $|\omega|$ approaches the limit $\varepsilon_2 \pi/2$, as seen in the optimal controller given by
\begin{subequations}\label{eq:u*_ln}
\begin{align}
v^* &= -\rho\varepsilon_1\arctan\left(\varepsilon_1L_{\bar{g}_1}V\right)\\
\omega^* &= -\varepsilon_2\arctan\left(\varepsilon_2L_{g_2}V\right)\,.
\end{align}
\end{subequations}
While the upper bound $|\omega| \leq \overline{\omega}$ is easily enforced by selecting $\varepsilon_2 = 2\overline{\omega}/\pi$, the upper bound $|v| \leq \overline{v}$ is more subtle. This is because the limiting value $\rho \varepsilon_1 \pi/2$ depends on $\rho$. To address this, we define
\begin{align}
\label{eq-eps1-rho}
\varepsilon_1(\rho) = \frac{\overline{v}}{\sigma + \rho} \frac{2}{\pi} > 0\,,
\end{align}
where $\sigma > 0$ is a small constant, to ensure the bound holds. Fig.~\ref{fig:ioc_controleffort_ex3} illustrates the result for $\bar{v} = \bar{\omega} = 1$ and $\sigma = 0.3$, showing that neither control inputs exceeds magnitude $1$.

The running cost on the state is then $\ell\eta(r) = r \arctan(r) - \frac{1}{2}\ln(1+r^2)\approx \frac{\pi}{2}r$ as $r \rightarrow \infty$ indicating that the state incurs a linearly increasing penalty for large values of $L_{\bar{g}_1}V$ and $L_{g_2}V$, while remaining relatively tolerant near $0$. 
\end{example}


\begin{figure}[t]
\centering
\begin{subfigure}{0.5\textwidth}
\centering
\includegraphics[width=.9\textwidth]{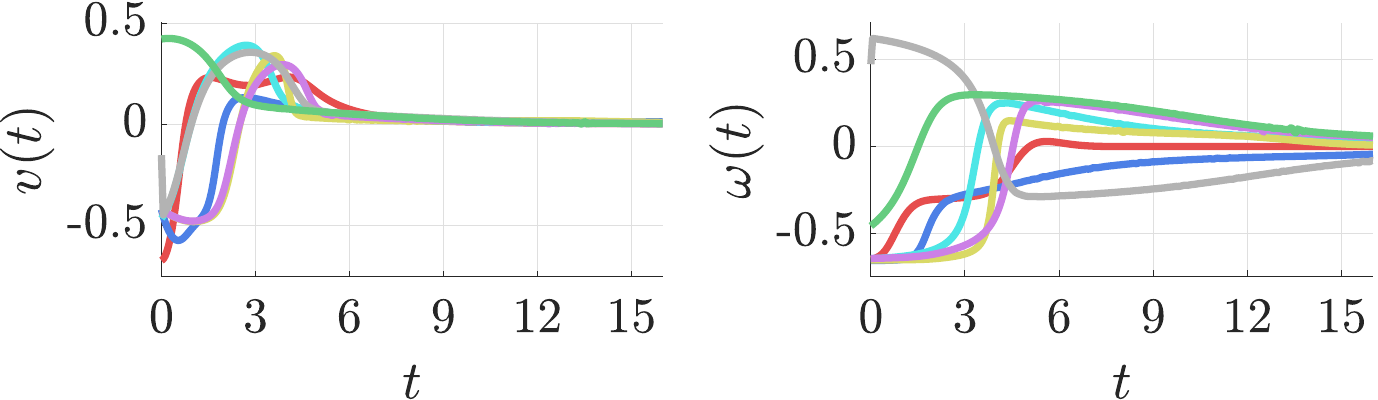}
\caption{Control effort of the bounded controller \eqref{eq:u*_ln}.}
\label{fig:ioc_controleffort_ex3}
\end{subfigure}
\vspace{0.25cm}
\begin{subfigure}{0.5\textwidth}
\centering
\includegraphics[width=.9\textwidth]{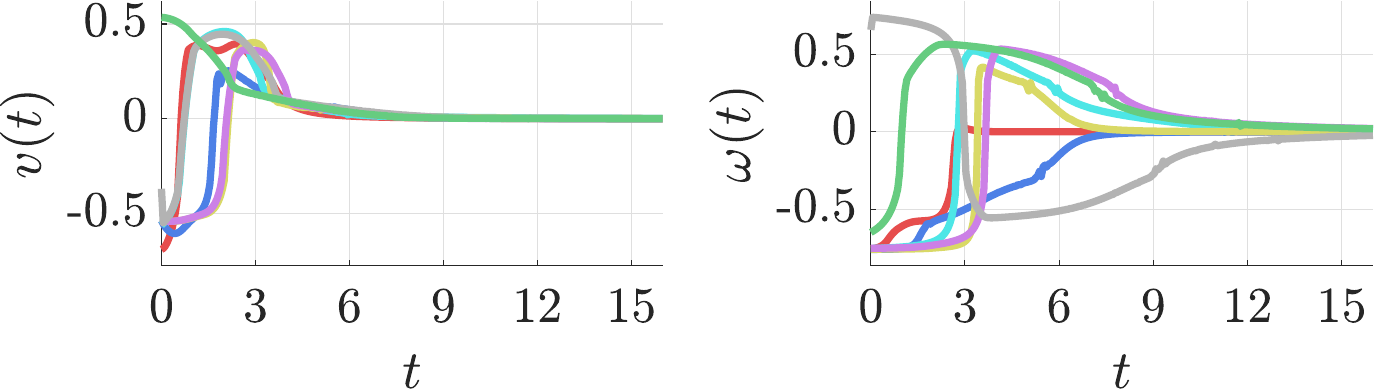}
\caption{Control effort of the bounded controller \eqref{eq:u*_bangbang}.}
\label{fig:ioc_controleffort_ex4}
\end{subfigure}
\caption{Comparison in control effort between Example~\ref{example:IOC3} and Example~\ref{example:IOC4}  with $\bar v = \bar \omega = 1$. Each color denotes the same corresponding initial conditions as in Fig.~\ref{fig:ioc_trajectory}.}
\vspace{-0.4cm}
\label{fig:ioc_ex3_vs_ex4}
\end{figure}

\begin{example}\label{example:IOC4}\textbf{(``Relay-approximating'' bounded control)}
Alternative to Example~\ref{example:IOC3}, consider the non-quadratic cost function defined by
\begin{align}\label{eq:eta_bangbang}
\eta(r) &= \frac{e}{e^{1/r} - e}
\,,
\end{align}
which represents the control penalty, is infinitely flat near the origin and blows up at $r = 1$. This means that small control magnitudes are penalized very lightly, while magnitudes approaching 1 incur an infinite cost, effectively enforcing a hard input constraint.

The state penalty $\ell\eta(r)$, while only numerically accessible, has infinite slope at $r = 0$, indicating that any nonzero value of the state incurs a large cost. However, the function grows slowly for large arguments, revealing a cost structure that strongly discourages small errors but is relatively tolerant of large deviations.

The associated derivative inverse $ (\eta')^{-1}(r) = \frac{1}{1+\ln\left(1+\frac{1}{r}\right)}$ is continuous, monotonically increasing, bounded by 1, and has an infinite derivative at the origin. It closely resembles a smoothed approximation of the discontinuous signum function commonly used in sliding mode control. This behavior, characterized by high gain for small arguments and saturation for large ones, is typical of feedback laws arising in time-optimal control. Although the cost here is not explicitly time-optimal, it captures similar characteristics, including finite-time convergence under bounded input. Hence, the resulting optimal feedback law is given as
\begin{subequations}\label{eq:u*_bangbang}
\begin{align}
v^* = -\rho\varepsilon_1\frac{\sgn(L_{\bar{g}_1}V)}{1+\ln\left(1+ \dfrac{1}{\varepsilon_1 |L_{\bar{g}_1}V|}\right)}\\
\omega^* = -\varepsilon_2\frac{\sgn(L_{g_2}V)}{1+\ln\left(1+ \dfrac{1}{\varepsilon_2 |L_{g_2}V|}\right)}\,.
\end{align}
\end{subequations}
However, for similar reasons as in Example~\ref{example:IOC3}, to guarantee the upper bound $|v| \leq \overline{v}$ and $|\omega| \leq \overline{\omega}$, we choose
\begin{align}
\varepsilon_1(\rho) &= \frac{\overline{v}}{\sigma + \rho} > 0\,,
\end{align}
for small $\sigma > 0$ and $\varepsilon_2 = \overline{\omega}$. Fig.~\ref{fig:ioc_controleffort_ex4} shows the simulation with $\bar{v} = \bar{\omega} = 1$ and $\sigma = 0.3$. Notably, unlike the concentrated, high initial effort observed in Example~\ref{example:IOC3}, the control law \eqref{eq:u*_bangbang} produces a more evenly sustained control effort over time. The settling times on the controls in Fig.~\ref{fig:ioc_controleffort_ex3} and Fig.~\ref{fig:ioc_controleffort_ex4} are different due to this more sustained control effort. This behavior is due to the cost function \eqref{eq:eta_bangbang}, which assigns minimal cost to small inputs, thereby tolerating low-magnitude control efforts over the trajectory. One can certainly not expect the achievement of fixed-time (FxT) stabilization under bounded control, but only finite-time (FT) stabilization. 
\end{example}

\begin{figure}[t]
\centering
\includegraphics[width=.7\linewidth]{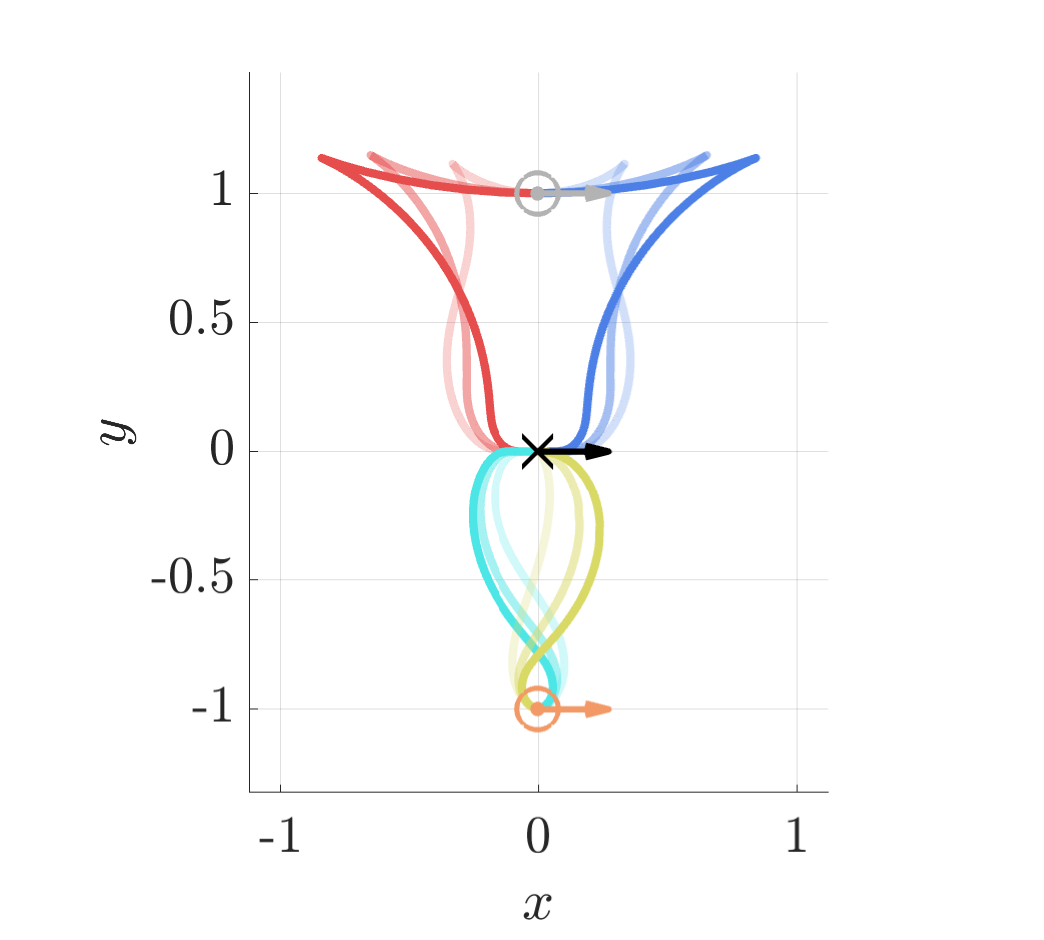}
\caption{Parallel parking trajectories for forward $L_gV$ control (red), reverse $L_gV$ control (blue), forward GloBa control~\cite[Thm. 7]{Part1_todorovskiCLF2025} (cyan), and reverse BoLSA control~\cite[Thm. 2]{Part1_todorovskiCLF2025} (yellow) with opacity representing different gain choices.}
\vspace{-0.4cm}
\label{fig:parallel_ops}
\end{figure}

\paragraph{Gain margin and robustness to actuator saturation} Optimality is a benefit of inverse optimality. But the key practical benefit of the $L_gV$ inverse optimal control is the {\em gain margin}, i.e., the robustness to the uncertainty in the input coefficients. In the unicycle model, the coefficients on the inputs $(v,\omega)$ are taken, implicitly, as $(1,1)$. In the simplest of the inverse optimal controllers \eqref{eq-LQ-forward}, the gains $\varepsilon_1,\varepsilon_2>0$ are arbitrary, which means that the inverse optimal controller has a gain margin of $(0,\infty)$. This gain margin is wider than the conventional $[1/2,\infty)$ of optimal controllers for systems with drift, which permits only halving the optimal gain without losing stability. The unusually strong margin $(0,\infty)$, which permits reducing the gain to an {\em arbitrarily low} value, follows from the unicycle being driftless.

Inverse optimality not only affords robustness to the uncertainty in the input coefficient. It also enables operation under arbitrarily low input saturation, as we observe in Examples~\ref{example:IOC3} and \ref{example:IOC4}. For example, taking $v^*$ and $\omega^*$ as in \eqref{eq:u*_ln}, \eqref{eq-eps1-rho}, we get the  feedback
\begin{align}
v^* &= -\bar v \ \underbrace{\frac{\rho}{\sigma + \rho} \frac{2}{\pi}\arctan\left(\frac{\bar v}{\sigma + \rho}\frac{2}{\pi} \textcolor{red}{L_{\bar{g}_1}V}
\right)}_{\mbox{\small magnitude $<1$}}\\
\omega^* &= -\bar\omega\ \underbrace{\frac{2}{\pi}\arctan\left(\bar\omega\frac{2}{\pi}\textcolor{red}{L_{g_2}V}
\right)}_{\mbox{\small magnitude $<1$}}\,,
\end{align}
where $\bar v, \bar\omega$ are the maximum forward and angular velocities, respectively, which the vehicle's propulsion and steering actuators can deliver, which may be arbitrarily low. So, the feedback always acts in the right direction, which is determined by the sign of the expressions in red under the $\arctan$. The $L_{\bar{g}_1}V$ and $L_{g_2}V$ terms in red (with full expressions given in~\eqref{eq-u*-quad0}) always deliver the ``correct directions'' for control by design: they are made to do so by our construction of the CLF $V(\rho,\delta,\gamma)$. 

\paragraph{Nonholonomic parallel parking} Parallel parking is a benchmark maneuver for nonholonomic vehicles, requiring precise coordination of orientation and position. With a wide range of control laws and infinitely many parameters across this paper and~\cite{Part1_todorovskiCLF2025}, our framework can be tailored to a given task. Fig.~\ref{fig:parallel_ops} shows several representative examples. The $L_gV$ controllers closely resemble human parallel parking, with the reverse $L_gV$ controller reproducing the traditional approach used by drivers. In contrast, the standard GloBa and reversed BoLSA controllers~\cite[Thms. 2 and 7]{Part1_todorovskiCLF2025} produce smoother ``S''-shaped trajectories better suited to tighter spaces.


\section{Adaptive Stabilization under Unknown Input Coefficients}\label{sec:adapt}

In view of the infinite gain margin as well as the $L_gV$ form feedback control law, in this section, we develop an adaptive control law for the unicycle model with model uncertainties.

Consider the model
\begin{subequations}
\label{eq:unicycle_polar_closed_loop-Gv-slip}
\begin{align}
\dot{\rho} &=  -b_1v \cos\gamma\\
\dot{\delta} &=  b_1\frac{v}{\rho}\sin(\gamma)
\\
\dot{\gamma} &= b_1\frac{v}{\rho}\sin(\gamma)
-b_2 \omega
\,,
\end{align}
\end{subequations}
where the constants $b_1,b_2$ are positive but otherwise completely unknown. For example, $b_1,b_2 \in (0,1]$ may physically represent {\em unknown} wheel slippage coefficients (on a two-wheeled mobile robot). 

Recall from \eqref{eq-u*-quad0} that
\begin{subequations}\label{eq-u*-quad0a}  
\begin{align}
\label{eq-LQ-forwarda}
L_{\bar{g}_1}V &:=  -\dfrac{\partial V}{\partial\rho}\rho\cos\gamma + \left(\dfrac{\partial V}{\partial\delta}+\dfrac{\partial V}{\partial\gamma}\right) \sin\gamma \\
\label{eq-LQ-steera}
L_{g_2}V &:= -\dfrac{\partial V}{\partial\gamma}\,,
\end{align}
\end{subequations}
and introduce the adaptive control laws
\begin{eqnarray}
\label{eq-v-adapt}
v &=& -\hat\varepsilon_1 \rho L_{\bar{g}_1}V\\
\label{eq-om-adapt}
\omega &=& -\hat\varepsilon_2 L_{g_2}V\,,
\end{eqnarray}
where $\hat \varepsilon_1, \hat \varepsilon_2$ are the online estimates of $1/b_1,1/b_2$, respectively. 
Take any of our strict CLFs $V$, on any of the state spaces considered. The derivative of $V$ is
\begin{equation}
\dot V = -(L_{\bar g_1} V)^2 \left(1 - b_1 \tilde\varepsilon_1 \right) -(L_{ g_2} V)^2 \left(1 - b_2 \tilde\varepsilon_2 \right)\,,
\end{equation}
where $\tilde\varepsilon_i = 1/b_i - \hat\varepsilon_i$. 
Taking the adaptive CLF
\begin{equation}\label{eq:CLF_adapt}
V_{\rm a} = \ln(1+n(V)) + \frac{b_1}{2\mu_1}\tilde\varepsilon_1^2 + \frac{b_2}{2\mu_2}\tilde\varepsilon_2^2 \,,
\end{equation}
with adaptation gains $\mu_1,\mu_2>0$ and any normalization function $n\in\mathcal{K}_\infty\cap C^1$, including, for example, $n(V) = n_0V, n_0>0$, we get
\begin{align}
\dot V_{\rm a} =& -\frac{n'(V)}{1+n(V)}\left[(L_{\bar g_1} V)^2 +(L_{ g_2} V )^2 \right]
\nonumber\\
& +b_1\tilde\varepsilon_1\left[  \frac{n'(V)}{1+n(V)}(L_{\bar g_1} V)^2  - \frac{\dot{\hat\varepsilon}_1}{\mu_1}\right] \nonumber\\
& +b_2\tilde\varepsilon_2\left[  \frac{n'(V)}{1+n(V)}(L_{ g_2} V)^2  - \frac{\dot{\hat\varepsilon}_2}{\mu_2}\right]\,.
\end{align}
Hence, we pick the update laws
\begin{eqnarray}
\label{eq-eps1-adapt}
\dot{\hat\varepsilon}_1 & = & \mu_1 \frac{n'(V)}{1+n(V)}(L_{\bar g_1} V)^2 \\
\label{eq-eps2-adapt}
\dot{\hat\varepsilon}_2 & = & \mu_2 \frac{n'(V)}{1+n(V)}(L_{ g_2} V)^2  \,, 
\end{eqnarray}
and obtain
\begin{equation}\label{eq:adapt_dotV}
\dot V_{\rm a} =  -\frac{n'(V)}{1+n(V)}\left[(L_{\bar g_1} V)^2 +(L_{ g_2} V )^2 \right]  \,,
\end{equation}
which is negative for all $(\rho,\delta,\gamma)\neq (0,0,0)$ on the state space considered but is not negative definite in the overall state of the adaptive system, which also includes $(\tilde\varepsilon_1,\tilde\varepsilon_2)$. With LaSalle's theorem, we obtain the following result.


\begin{theorem} \label{thrm:adapt}
Consider the system \eqref{eq:unicycle_polar_closed_loop-Gv-slip} with arbitrary unknown $b_1,b_2>0$, along with the control law \eqref{eq-v-adapt}, \eqref{eq-om-adapt} and the update laws \eqref{eq-eps1-adapt}, \eqref{eq-eps2-adapt}. For each CLF $V$ given by~\cite[Thms. 1-8]{Part1_todorovskiCLF2025}, with class $\mathcal{K}_\infty$ functions $\alpha_1,\alpha_2$ such that 
\begin{align}
\alpha_1\left(|(\rho,\delta,\gamma)|_{\mathcal{Q}}\right) \leq V(\rho,\delta,\gamma)
\leq \alpha_2\left(|(\rho,\delta,\gamma)|_{\mathcal{Q}}\right)\,,
\end{align}
for all initial conditions $(\rho(0),\delta(0),\gamma(0))$ on its respective state space $\mathcal{Q} \in \{\mathcal{S}_0,\mathcal{S}_1,\mathcal{S}_2,\mathcal{S}_3\}$, and for all parameter initial conditions $\hat \varepsilon_1(0)\in\mathbb{R},\hat \varepsilon_2(0)\in\mathbb{R}$, the following holds:
\begin{align}\label{eq:adapt_upsilon}
&\Upsilon(t) \leq a_1^{-1}\left( M\left(e^{m \, a_2(\Upsilon(0))} - 1\right)\right)\,,\quad \forall t \geq 0\,,
\end{align}
with
\begin{align}
\Upsilon &= 
|(\rho,\delta,\gamma)|_{\mathcal{Q}} +|(\tilde\varepsilon_1,\tilde\varepsilon_2)|\,,
\end{align}
where
\begin{eqnarray}
a_1(r) &=& \min\{n\circ\alpha_1(r),r^2\}\label{eq:adapt_a1}
\\
a_2(r) &=& \max\{n\circ\alpha_2(r),r^2\}
\end{eqnarray}
are $\mathcal{K}_\infty$ functions and
\begin{align}
M =  
\max\left\{1,\dfrac{1}{c_1}\right\}
\,, & \quad
\label{eq:adapt_M}
m =  1+c_2
\,,
\\
\label{eq-adaptivec1c2}
c_1 =  \min\left\{\frac{b_1}{2\mu_1},\frac{b_2}{2\mu_2}\right\}
\,, &\quad
c_2 =  \max\left\{\frac{b_1}{2\mu_1},\frac{b_2}{2\mu_2}\right\}\,.
\end{align}
In addition, $\rho(t),\delta(t),\gamma(t)\rightarrow 0$ as $t\rightarrow\infty$. 
\end{theorem}

\begin{proof}
For all $t \geq 0$, from~\eqref{eq:adapt_dotV} it follows that $V_a(t) \leq V_a(0)$, and from~\eqref{eq:CLF_adapt} that 
\begin{align}
&\ln\bigl(1+n\left(V(t)\right)\bigr) + c_1\left(\tilde{\varepsilon}_1^2(t) + \tilde{\varepsilon}_2^2(t)\right)\nonumber\\
&\leq\ln\bigl(1+n\left(V(0)\right)\bigr) + c_2\left(\tilde{\varepsilon}_1^2(0) + \tilde{\varepsilon}_2^2(0)\right)\,.
\end{align}
Exponentiating both sides of the inequality, using the inequality ${\rm e}^{s} \geq 1+s$ leads to \eqref{eq:adapt_upsilon} with \eqref{eq:adapt_a1}--\eqref{eq-adaptivec1c2}. The convergence follows from \eqref{eq:adapt_dotV} with LaSalle's theorem.
\end{proof}

\begin{figure}[t!]
\centering
\begin{subfigure}[b]{\linewidth}
\centering
\includegraphics[width=.75\linewidth]{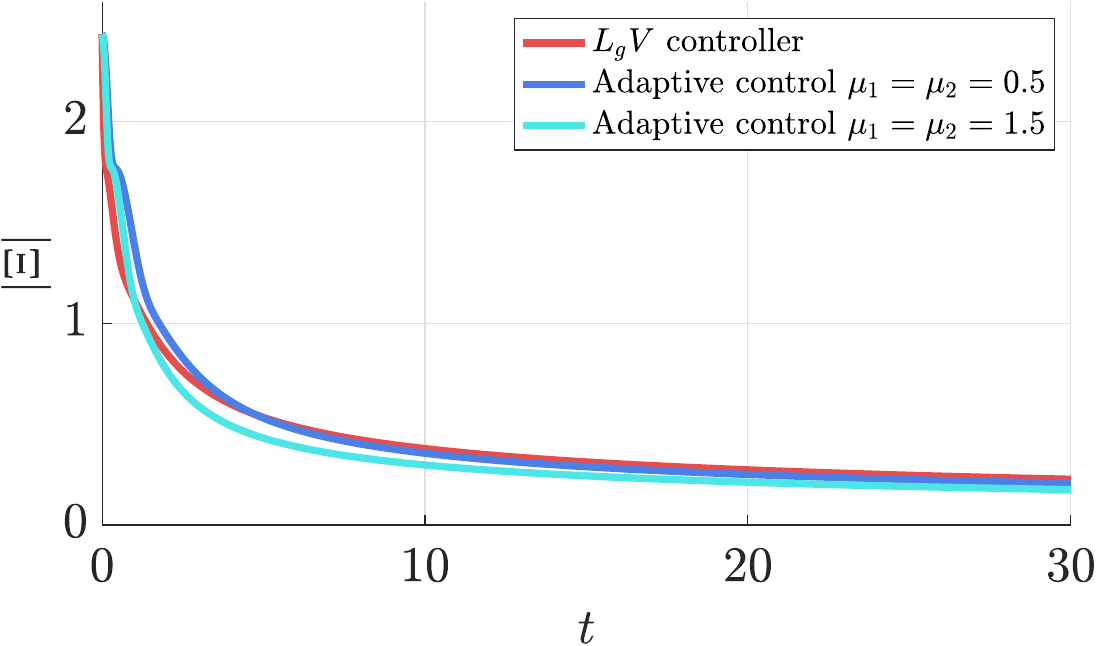}
\caption{Norm of $\Xi \coloneqq (\rho,\delta,\gamma)$ over time for each control law showing faster decay rate for both adaptive control laws.}
\label{fig:adapt_state_norm}
\end{subfigure}
\hfill
\vspace{0.1cm}
\begin{subfigure}[b]{\linewidth}
\centering
\includegraphics[width=.85\linewidth]{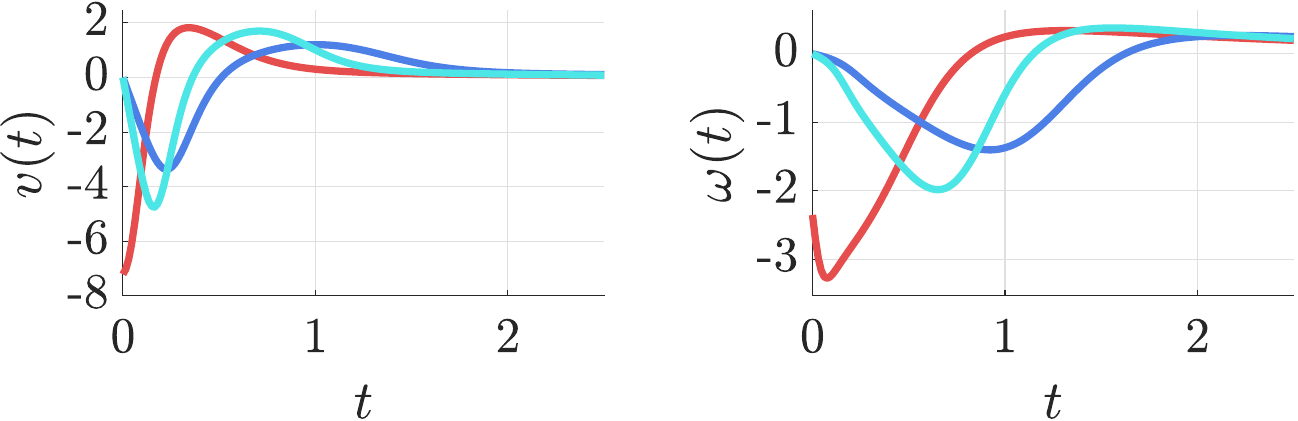}
\caption{Maximum magnitude of the control input is smaller for both adaptive control laws compared to the $L_gV$ controller despite the faster decay rate.}
\label{fig:normalized_control_input}
\end{subfigure}
\hfill
\vspace{0.1cm}
\begin{subfigure}[b]{\linewidth}
\centering
\includegraphics[width=.85\linewidth]{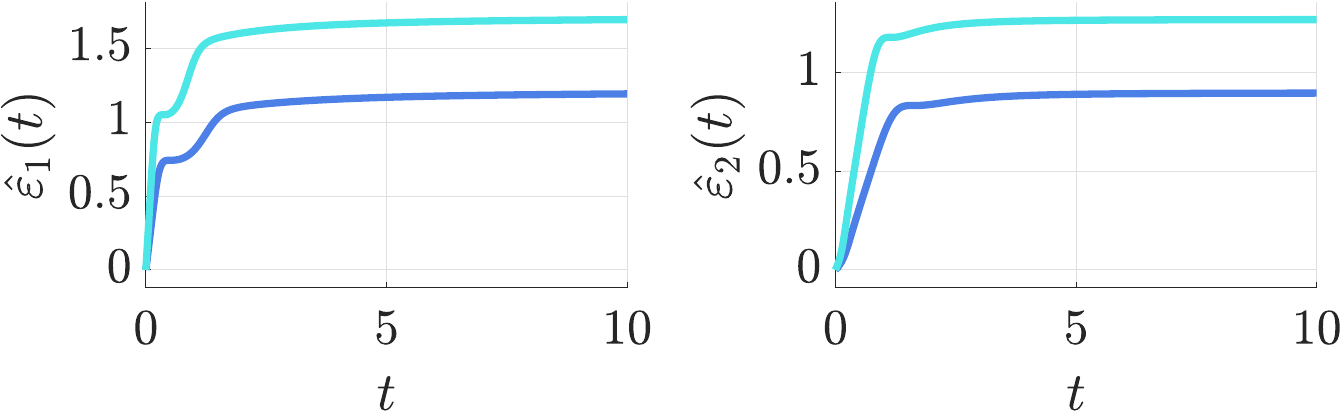}
\caption{$\hat{\varepsilon}_1$ and $\hat{\varepsilon}_2$ grows above $1$ leading to the faster decay rate.}
\label{fig:adaptive_eps}
\end{subfigure}
\caption{Comparison of state trajectories and control inputs under the quadratic cost $L_gV$ controller with both $b_1\varepsilon_1$ and $b_2\varepsilon_2$ taken as $1$ and the proposed adaptive controller with two different adaptation gains.}
\vspace{-0.3cm}
\label{fig:adaptive_control_fig}
\end{figure}

To illustrate the performance of the proposed method, we consider the BoFo CLF $V = \rho^2 + V_{\delta\gamma}$, where $V_{\delta\gamma}$ is defined in~\cite[Thm. 6]{Part1_todorovskiCLF2025}, and choose gains $k_1 = 3$, $k_2 = 3.5$, and $k_3 = 4$. In Fig.~\ref{fig:adaptive_control_fig}, we compare the  quadratic cost $L_gV$ controller (red), given by $v = -\rho L_{\bar{g}_1}V$ and $\omega = -L_{g_2}V$, with the adaptive control law in \eqref{eq-eps1-adapt}--\eqref{eq-eps2-adapt}, which uses a normalization function of the form $n(V) = V$, initial parameter estimates $\hat{\varepsilon}_1(0) = \hat{\varepsilon}_2(0) = 0$, and adaptation gains $\mu_1 = \mu_2 = 0.5$ (blue) and $\mu_1 = \mu_2 = 1$ (cyan). The system is initialized at $(\rho(0),\delta(0),\gamma(0)) = (1,-\pi/2,-\pi/2)$, under uncertainty in $b_1 = b_2 = 1$. For the quadratic cost $L_gV$ controller, this corresponds to the assumption of $b_1\varepsilon_1 = b_2\varepsilon_2 = 1$. The adaptive controller exhibited smaller peak values in both forward velocity and steering input (Figure~\ref{fig:normalized_control_input}), even while achieving a faster decay of the state norm (Figure~\ref{fig:adapt_state_norm}). This behavior is clearly due to $\hat{\varepsilon}_{1,2}$ starting small but quickly growing beyond $1$, thereby accelerating the decay rate to be faster than the $L_gV$ controller.

The adaptive control law 
\eqref{eq-v-adapt}, \eqref{eq-om-adapt}, \eqref{eq-eps1-adapt}, \eqref{eq-eps2-adapt}
has similarities with the inverse optimal controller \eqref{eq-u*-quad0}. The difference is that the adaptive controller works even if it starts with gains of the wrong sign, $\hat\varepsilon_1(0)<0, \hat\varepsilon_2(0)<0$. Additionally, the adaptive controller's gain is improved, online, by learning from the transients reflected in $L_{\bar g}V$.

Of course, all update laws exhibit drift in the presence of disturbances. By adding leakage to the update laws \eqref{eq-eps1-adapt} and \eqref{eq-eps2-adapt}, as is standard~\cite{ioannou1996robust}, boundedness of the vehicle states and parameter estimates, as well as practical regulation of $\rho(t),\delta(t),\gamma(t)$, would be ensured.

\section{Prescribed and Fixed-Time Nonholonomic Parking}\label{sec:PT}

In this section, we modify the controllers designed in Section~\ref{sec:nonmodular_unicycle_stabilization} such that the states $(\rho,\delta,\gamma)$ converge to the point $\rho = \delta = \gamma = 0$ in user-defined prescribed time $T > 0$. 

\begin{figure*}[t]
\centering
\begin{subfigure}{\linewidth}
\centering
\includegraphics[width=.87\linewidth]{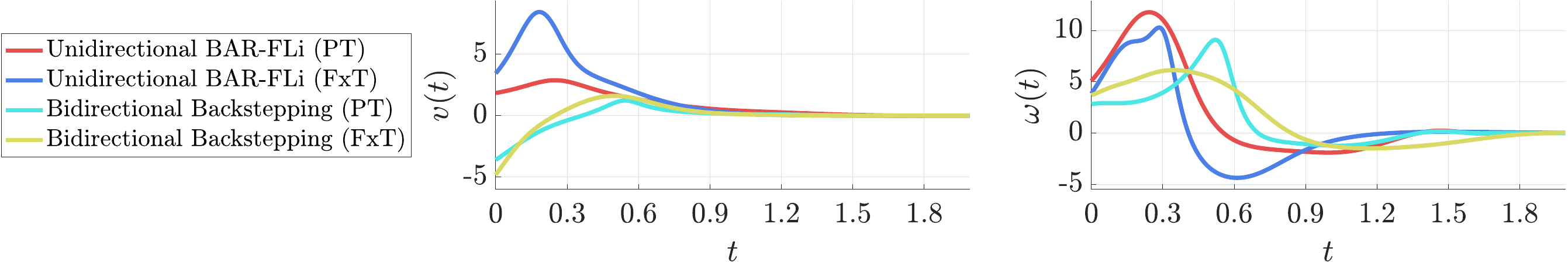}
\caption{Control inputs}
\label{fig:PT_control}
\end{subfigure}
\vspace{0.5em}  
\begin{subfigure}{\linewidth}
\centering
\includegraphics[width=.87\linewidth]{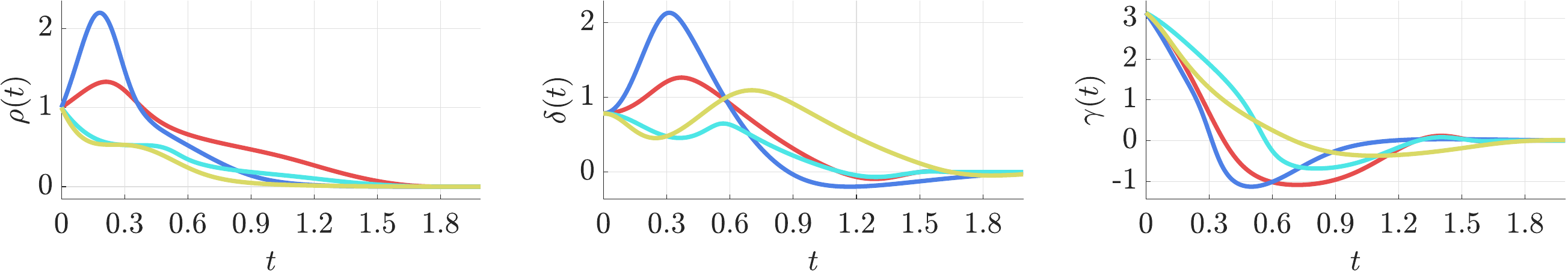}
\caption{State trajectories}
\label{fig:PT_traj}
\end{subfigure}


\caption{PT- and FxT-parking under~\eqref{eq:pt-control-laws} and~\eqref{eq:fxt-control-laws-exp} with $T=2$ ($p=0.3$ for FxT), applied to both the unidirectional BAR-FLi~\eqref{eq:unidir_control_v}--\eqref{eq:unidir_control_omegabar} and bidirectional backstepping~\eqref{eq:bidir_control_v}--\eqref{eq:bidir_control_omegabar} controllers. Gains $[k_1,k_2,k_3,k_4]=[1,2.2,2.5,0.5]$ (PT) and $[3,0.8,0.5,1]$ (FxT); they differ to give comparable behavior despite the distinct scaling of the two formulations.}

\vspace{-0.3cm}
\label{fig:PT}
\end{figure*}

\begin{theorem}[PT-Parking]
\label{thm:PT-control}
Consider the system \eqref{eq:unicycle_polar} in closed-loop with
\begin{subequations}
\label{eq:pt-control-laws}
\begin{align}
v(t) &= (1 + \nu^2(t-t_0))\tilde{v}(t)\label{eq:fxt-v}\\
\omega(t) &= (1 + \nu^2(t-t_0))\left( \frac{\tilde{v}}{\rho}\sin(\gamma) + \tilde{\omega}(t) \right) \label{eq:fxt-omega}
\end{align}
\end{subequations}
where 
\begin{equation}
\nu(t-t_0) = \tan\left(\frac{\pi(t-t_0)}{2T}\right)
\end{equation}
with $T > 0$ and the forward and angular velocities $\tilde{v}, \tilde{\omega}$ are respectively chosen from: (i) \eqref{eq:unidir_control_v} and \eqref{eq:unidir_control_omegabar}; or (ii) \eqref{eq:bidir_control_v} and \eqref{eq:bidir_control_omegabar}, each of which requires the gains to satisfy $k_1, k_2, k_3, k_4 > 0$.
Then, 
\begin{enumerate}
\item \label{it:fxt-1} there exists $K, \lambda >0$, such that $|(\rho(t),\delta(t), \gamma(t))|_{\mathcal{Q}}\leq K|(\rho_0,\delta_0, \gamma_0)|_{\mathcal{Q}} {\rm e}^{-\frac{2T}{\pi}\lambda \nu(t-t_0)} $  for all $t\in[t_0,t_0+T)$, where the state-space $\mathcal{Q}$ corresponds to any of $\mathcal{S}_0, \mathcal{S}_2$ depending on the choice of $\tilde{v}$, $\tilde{\omega}$,
\item \label{it:fxt-2} and the input signals in the control laws \eqref{eq:pt-control-laws} are bounded and converge to zero as $t \rightarrow t_0 + T$.
\end{enumerate}
\end{theorem}

\begin{proof}
\textbf{Step 1:} We dilate the finite time interval $t \in [t_0, t_0 + T)$ to the infinite time interval $\tau \in [t_0, \infty)$ via the transformation $t \rightarrow \tau$:
\begin{equation}
\tau = t_0 + \frac{2T}{\pi} \nu(t-t_0). \label{eq:t_to_tau}
\end{equation}
Considering the fact that
\begin{equation}
\frac{{\rm d} t}{{\rm d} \tau} = \frac{1}{1 + \left(\frac{\pi}{2T}(\tau - t_0)\right)^2} \,,
\end{equation}
the control laws \eqref{eq:pt-control-laws} in time $\tau$ are 
\begin{subequations}
\label{eq:fxt-control-laws-in-tau}
\begin{align}
v(\tau) &= \left(1 + \frac{\pi^2}{4T^2}(\tau - t_0)^2\right) \tilde{v}, \\
\omega(\tau) &= \left(1 + \frac{\pi^2}{4T^2}(\tau - t_0)^2\right)\left(\frac{\tilde{v}}{\rho}\sin(\gamma) + \tilde{\omega} \right)\,.
\end{align}
\end{subequations}

Hence, the closed-loops in time $\tau$ are exactly the same as \eqref{eq:rho_dot_unidir}, \eqref{eq:delta_dot_unidir}, \eqref{eq:z_dot_unidir} with \eqref{eq:unidir_control_v} and \eqref{eq:unidir_control_omegabar} or as \eqref{eq:rho_dot_bidir}, \eqref{eq:delta_dot_bidir},  \eqref{eq:z_dot_bidir} with \eqref{eq:bidir_control_v} and \eqref{eq:bidir_control_omegabar} where ${\rm d}/ {\rm d} t$ replaced by ${\rm d}/ {\rm d} \tau$.
Thus, the corresponding stability guarantees stated in Theorems~\ref{thm:unidirectional_bkstp}--\ref{thm:bidirectional_bkstp} hold in time $\tau$, respectively. 
This implies the existence of positive constants $K, \lambda > 0$, such that, 
\begin{equation}
\hspace*{-0.38cm}
\abs{(\rho(\tau),\delta(\tau),\gamma(\tau))}_{\mathcal{Q}} \le K \abs{(\rho(t_0), \delta(t_0),\gamma(t_0))}_{\mathcal{Q}} {\rm e}^{-\lambda(\tau - t_0)} \label{eq:KL-estimate-tau}
\end{equation}
where $\mathcal{Q}$ corresponds to any of $\mathcal{S}_0,  \mathcal{S}_2$ depending on the choice of $\tilde{v}$ and $\tilde{\omega}$.
In addition, we  bound the control laws \eqref{eq:fxt-control-laws-in-tau} as
\begin{equation}
\begin{aligned}[b]
\abs{v(\tau)} + \abs{\omega(\tau)} \le &\left(1 + \frac{\pi^2}{4T^2}(\tau - t_0)^2\right) \times \\ &\tilde{\alpha}(\abs{(\rho(\tau),\delta(\tau),\gamma(\tau))}_{\mathcal{Q}}), 
\end{aligned}
\label{eq:norm_control_law_tau}
\end{equation}
where $\tilde{\alpha} \in \mathcal{K}_{\infty}$ is defined according to the chosen controller and the corresponding state space $\mathcal{Q}$ as follows: On $\mathcal{Q} = \mathcal{S}_0$ the function $\tilde{\alpha}(r) = \max\{c_1, c_2\} (r + r^2 + r^3)$ corresponds to  \eqref{eq:bidir_control_v}, \eqref{eq:bidir_control_omegabar} with $c_1 = 2k_1 \sqrt{1 + k_2^2} + k_4 (1+k_2) + k_3 \sqrt{1 + k_2^2} + k_1 k_2 (\sqrt{1 + k_2^2}(1+k_2) + 1)$ and $c_2 = 4k_1k_2 + k_3 + 4k_2 k_3 + 2k_1 k_2^2(1 + k_2)$, while
on $\mathcal{Q} =\mathcal{S}_2$, the function $\tilde{\alpha}(r) = c_1 (r+r^2 + r^3)$ corresponds to \eqref{eq:unidir_control_v}, \eqref{eq:unidir_control_omegabar}.
Based on \eqref{eq:KL-estimate-tau} and \eqref{eq:norm_control_law_tau}, we obtain
\begin{equation}
\begin{aligned}[b]
\abs{v(\tau)} + &\abs{\omega(\tau)} \le \left(1 + \frac{\pi^2}{4T^2}(\tau - t_0)^2\right) \times  \\ &\tilde{\alpha}\left( K \abs{\left(\rho(t_0), \delta(t_0), \gamma(t_0)\right)}_{\mathcal{Q}} {\rm e}^{-\lambda (\tau-t_0)} \right). \label{eq:input_exp_bound_tau}
\end{aligned}
\end{equation}
Considering the functions $\tilde{\alpha}$ as in \eqref{eq:norm_control_law_tau}, established above to all satisfy $\tilde\alpha(r)\leq a(r+r^2 + r^3), a\in\mathbb{R}_{>0}$,  and letting $\tau \rightarrow \infty$ in \eqref{eq:input_exp_bound_tau}, we observe that $\abs{v(\tau)} + \abs{\omega(\tau)}$ converges to zero.

\textbf{Step 2:} We now convert the above results into time $t$ by contracting the infinite time interval $[t_0, \infty)$ into the finite time interval $[t_0, t_0 + T)$ in time $t$, by taking the inverse transformation of \eqref{eq:t_to_tau} as
\begin{equation}
t = t_0 + \frac{2T}{\pi} \arctan\left(\frac{\pi}{2T} (\tau - t_0)\right).
\end{equation}
Then, \eqref{eq:KL-estimate-tau} in time $t$ reads 
\begin{equation}
\begin{aligned}[b]
\abs{(\rho(t),\delta(t),\gamma(t))}_{\mathcal{Q}} \le K\abs{(\rho(t_0), \delta(t_0),\gamma(t_0))}_{\mathcal{Q}}  {\rm e}^{-\frac{2T}{\pi}\lambda \nu(t-t_0)} 
\end{aligned}
\end{equation}
which proves property~\ref{it:fxt-1} of Theorem~\ref{thm:PT-control}. 
Analogously, \eqref{eq:input_exp_bound_tau} in time $t$ reads as
\begin{equation}
\begin{aligned}[b]
&\abs{v(t)} + \abs{\omega(t)}  \le (1 + \nu^2(t-t_0)) \times \\
&\tilde{\alpha}\left( K\abs{(\rho(t_0), \delta(t_0),\gamma(t_0))}_{\mathcal{Q}}  {\rm e}^{-\frac{2T}{\pi}\lambda \nu(t-t_0)} \right). \label{eq:input_exp_bound_t}
\end{aligned}
\end{equation}
Letting $t \rightarrow t_0 + T$ in \eqref{eq:input_exp_bound_t}, we conclude that $\abs{v(t)} + \abs{\omega(t)}$ is bounded and converges to zero. Thus, property~\ref{it:fxt-2} of Theorem~\ref{thm:PT-control} holds.
\end{proof}

\begin{theorem}[FxT-Parking]
\label{thm:FxT-control}
Consider the system \eqref{eq:unicycle_polar} in closed-loop with
\begin{subequations}
\label{eq:fxt-control-laws-exp}
\begin{align}
v(t) &=\kappa(\rho,\delta,\gamma)\tilde{v}(t)\label{eq:fxt-v-exp}\\
\omega(t) &= \kappa(\rho,\delta,\gamma) \left( \frac{\tilde{v}}{\rho}\sin(\gamma) + \tilde{\omega}(t) \right) \label{eq:fxt-omega-exp}
\end{align}
\end{subequations}
where 
\begin{equation}
\kappa(\rho,\delta,\gamma) = \frac{1}{\underline{c}   pT} \exp(V^p) V^{-p}, 
\end{equation}
with $p \in (0, \frac{1}{2})$ and $T > 0$ and the forward and angular velocities $\tilde{v}, \tilde{\omega}$, the backstepping transformations $z$ and the Lyapunov functions $V$ are respectively chosen from: (i) \eqref{eq:unidir_control_v}, \eqref{eq:unidir_control_omegabar}, \eqref{eq:unidir_z} and \eqref{eq:unidir_V_general} with $\underline{c}$ as in \eqref{eq:unidir_Vdot1}; or (ii) \eqref{eq:bidir_control_v},  \eqref{eq:bidir_control_omegabar}, \eqref{eq:bidir_z} and \eqref{eq:bidir_V_general} with $\underline{c}$ as in \eqref{eq:bidir_Vdot1}, each of which requires the gains to satisfy $k_1, k_2, k_3, k_4 > 0$.
Then, 
\begin{enumerate}
\item \label{it:fxt-1-exp} there exists $K_1, K_2 >0$, such  that it holds
\begin{equation} 
\hspace*{-0.4cm}
\begin{aligned}[b]
|(&\rho(t),\delta(t), \gamma(t))|_{\mathcal{Q}} \leq \\ &K_1 \left[\ln \left( \frac{1}{\frac{t-t_0}{T} + \exp(-K_2 \abs{(\rho_0,\delta_0, \gamma_0)}^{2p}_{\mathcal{Q}})}\right)\right]^{\frac{1}{2p}} \label{eq:fxt_states_bound}
\end{aligned}
\end{equation}  
for all $t\in[t_0, t_0 + T_{\rm s}(\rho_0, \delta_0, \gamma_0))$, and  $|(\rho(t),\delta(t), \gamma(t))|_{\mathcal{Q}} = 0$ for all $t \ge t_0 + T_{\rm s}(\rho_0, \delta_0, \gamma_0)$, and for all $(\rho_0,\delta_0,\gamma_0)$, the settling time $T_{\rm s}$ is such that
\begin{equation}
\hspace*{-0.2cm}
  T_{\rm s}(\rho_0,\delta_0, \gamma_0) \le [1 - \exp(-K_2 \abs{(\rho_0,\delta_0,\gamma_0)}_{\mathcal{Q}}^{2p})] T  \leq T \,, \label{eq:settling_time}
\end{equation}
where the state-space $\mathcal{Q}$ corresponds to either $\mathcal{S}_0$ or $ \mathcal{S}_2$ depending on the choice of $\tilde{v}$ and $\tilde{\omega}$,
\item \label{it:fxt-2-exp} and the input signals in the control laws \eqref{eq:fxt-control-laws-exp} are bounded and converge to zero no later than $t = t_0 + T$.
\end{enumerate}
\end{theorem}

\begin{proof}
With the rescaled control laws \eqref{eq:fxt-control-laws-exp}, we follow the proofs of Theorems~\ref{thm:unidirectional_bkstp} and~\ref{thm:bidirectional_bkstp} and arrive at the rescaled versions of \eqref{eq:unidir_Vdot1} and \eqref{eq:bidir_Vdot1}, which now read as
\begin{equation}
   \dot V \le - \underline{c} k(\rho,\delta,\gamma) V  = - \frac{1}{p T} \exp(V^p) V^{1-p}.   
\end{equation}
By the comparison principle, this implies that
\begin{equation}
    V(t) \le \left[\ln\left(\frac{1}{\frac{t-t_0}{T} + \exp(-V_0 ^p)}\right)\right]^{1/p} \,.
\end{equation}
Considering the fact that both Lyapunov functions  \eqref{eq:unidir_V_general} and \eqref{eq:bidir_V_general} are bounded as
\begin{equation}
 \bar k_1 \abs{(\rho,\delta, \gamma)}_{\mathcal{Q}}^2 \le V \le  \bar k_2 \abs{(\rho,\delta,\gamma)}_{\mathcal{Q}}^2  \,,\label{eq:fxt_V_up_low_bound} 
\end{equation}
where $\bar k_1 = \min\left\{ 1/2, k_1/(2k_1 k_2^2 + k_3)\right\}$, $\bar k_2 = \max\left\{1, 1 + 2 k_1k_2^2/k_3, 2k_1/k_3\right\}$, on their respective state-space $\mathcal{Q} = \{\mathcal{S}_0, \mathcal{S}_2\}$, we arrive at \eqref{eq:fxt_states_bound}.
 The norm of the control laws \eqref{eq:fxt-control-laws-exp} is
\begin{equation}
    \abs{v} + \abs{\omega} \le  \frac{1}{\underline{c} p T} \exp(V^p) V^{-p} \tilde{\alpha}(\abs{(\rho,\delta,\gamma)}_{\mathcal{Q}})\,,
\end{equation}
where $\tilde{\alpha}(r) = a(r+r^2 + r^3), a \in \R_{>0}$ as already established in \eqref{eq:norm_control_law_tau}. 
Considering the lower bound from \eqref{eq:fxt_V_up_low_bound}, we obtain 
\begin{equation}
\begin{aligned}[b]
    \abs{v} + \abs{\omega} \le  &\frac{a\exp(V^p)}{\underline{c} \bar k_1^p p T} (\abs{(\rho,\delta,\gamma)}_{\mathcal{Q}}^{1-2p} \\ &  +\abs{(\rho,\delta,\gamma)}_{\mathcal{Q}}^{2-2p} + \abs{(\rho,\delta,\gamma)}_{\mathcal{Q}}^{3-2p})
\end{aligned}
\label{eq:norm_control_fxt}
\end{equation}
where $1-2p > 0$, $2-2p >0$,  $3-2p > 0$ since $p \in (0,1/2)$, hence \eqref{eq:norm_control_fxt} is bounded as $\abs{(\rho,\delta,\gamma)}_{\mathcal{Q}}$ goes to zero. Now, consider \eqref{eq:fxt_states_bound} and note that if $(t-t_0)/T + \exp(-K_2 \abs{(\rho,\delta,\gamma)}_{\mathcal{Q}}^{2p}) = 1$, then $\abs{(\rho,\delta,\gamma)}_{\mathcal{Q}}  = 0$. Since ``$\exp$'' is less than or equal to 1, the settling time $T_s(\rho_0,\delta_0, \gamma_0) \le T [1 - \exp(-K_2 \abs{(\rho_0,\delta_0,\gamma_0)}_{\mathcal{Q}}^{2p})] \le T$, from which statements \ref{it:fxt-1-exp} and \ref{it:fxt-2-exp} of Theorem \ref{thm:FxT-control} follow.
\end{proof}

The PT-controllers \eqref{eq:pt-control-laws} admit a directly tunable settling time, since the convergence horizon $T$ can be explicitly prescribed by the designer and remains independent of both the initial conditions and the controller gains.
In contrast, the more conservative FxT-controllers \eqref{eq:fxt-control-laws-exp} provide only an upper-bound estimate of the settling time of $(\rho,\delta,\gamma)$ to zero given in \eqref{eq:settling_time} that depends on the initial conditions and the chosen control parameters. Consequently, the designer can specify only an overestimate $T$ of the true convergence time.
 Moreover, the FxT-control laws \eqref{eq:fxt-control-laws-exp} are nonsmooth even in the polar coordinates, and when they are implemented in Cartesian form, the discontinuity of the coordinate transformation (see Table~\ref{tab:polar_coordinates}) is added on top of this nonsmoothness, making them more sensitive to measurement noise. As opposed to that, the PT-control laws 
 \eqref{eq:pt-control-laws} yield control signals that are smooth in both time and the states $(\rho,\delta,\gamma)$. 


The PT controller is defined only on the interval $[t_0, t_0 + T)$, which can be advantageous in applications that require the motion to be completed within a prescribed-time window, after which the control can be safely deactivated. In contrast, the FxT controller remains defined for all $t \geq t_0$, making it more suitable for ongoing tasks over an infinite horizon. In both cases, the control inputs remain bounded and converge to zero, ensuring comparable terminal behavior. The properties in Theorems~\ref{thm:PT-control} and~\ref{thm:FxT-control} are illustrated by the simulations in Fig.~\ref{fig:PT}.

\section{Nonovershooting Control: Strictly Forward Parking without Curb Violation}\label{sec:safety}

Finally, we incorporate an analytical nonovershooting design that enforces safety constraints. Suppose the goal is to regulate the unicycle to $x = y = \theta = 0$ without overshooting in $y$, i.e., with $y(t) \leq 0$ for all $t \geq t_0$, akin to parking a vehicle without going over the curb. Since $y = -\rho\sin\delta$ with $\delta \in [-\pi,\pi)$, this holds precisely for $\delta \in [0,\pi)$. We therefore pose ``nonovershooting control'' as the stabilization of $\rho=\delta=\gamma=0$ with arbitrary $\delta_0\in(0,\pi)$ and $\gamma_0$ possibly limited to an interval around $\gamma=0$.


We use the BAR-FLi CLFs of~\cite[Thm. 8]{Part1_todorovskiCLF2025}, whose barrier-like properties as in~\CHG{\cite[Sec.~VIII]{Part1_todorovskiCLF2025}} restrict the polar angle $\delta$. Building on the BAR-FLi backstepping transformation, we design $\tilde{\omega}(t)$ to keep $\delta(t) \in [0, \pi)$, and hence $y(t) \leq 0$, for all $t \geq 0$. The following result gives a control law that renders the origin GAS on $\mathcal{S}_2$ and, for suitable gains, keeps $\delta(t) \in [0,\pi)$ whenever $\delta_0 \in (0, \pi)$ and $\gamma_0 \in (-\pi/4, \pi/4)$. Under the same conditions, the forward velocity $v(t)$ remains nonnegative---crucial when the unicycle models fixed-wing aircraft or guided missiles. Fig.~\ref{fig:nonovershoot_fig} shows the admissible region in $(\delta,\gamma)$ space and its geometric interpretation in Cartesian coordinates.


\begin{theorem}\label{thrm:nonovershoot}
For the system \eqref{eq:unicycle_polar}, the control laws 
\begin{eqnarray}
v &=&  k_1 \rho \cos(\gamma)\label{eq-v-general}\\
\omega &=& \dfrac{k_1}{2} \sin(2\gamma) +\tilde\omega\label{eq-omega-general}\\
\tilde{\omega} &=& \left(k_4 + \dfrac{k_3}{k_2} 
\psi^2(z,\gamma)N(\delta)
\left(1+\tan^2\displaystyle\frac{\delta}{2}\right)\right) z \nonumber\\
&& + \dfrac{k_1}{2}\dfrac{k_2}{\cos^2\dfrac{\delta}{2}+ 16k_2^2\sin^2\dfrac{\delta}{2}} \sin(2\gamma)\label{eq:omega_nonovershoot}\,,
\end{eqnarray}
where 
\begin{eqnarray}
z &=& \gamma + \frac{1}{2}\arctan\left(4k_2\tan\frac{\delta}{2}\right)\label{eq:z_BARFLi}\\
N(\delta) &=& \sqrt{1 + 16 k_2^2 \tan^2(\delta/2)}\label{eq:N_BARFLi}\\
\psi(z,\gamma) &=& \frac{\sin(2z-2\gamma) +\sin(2\gamma)}{2z}\label{eq:psi_definition}\,,
\end{eqnarray}
and with $k_1, k_2, k_3, k_4 > 0$, render the point $\rho = \delta = \gamma = 0$ GAS on $\mathcal{S}_2$, in accordance with Definition~\ref{def-our-GAS}. Moreover, if the initial conditions satisfy $\delta_0 \in (0, \pi)$ and $\gamma_0 \in (-\pi/4, \pi/4)$ with $k_2$ chosen as

\begin{align}\label{eq:nonovershoot_k2_interval}
k_2 \in \left(\frac{\tan\left(\max\{0,-2\gamma_0\}\right)}{4\tan(\delta_0/2)},\frac{\tan\left(\frac{\pi}{2}+\min\{0,-2\gamma_0\}\right)}{4\tan(\delta_0/2)}\right)
\,,
\end{align}
the polar angle $\delta(t)$ remains in the interval $[0, \pi)$, implying that $(x(t),y(t),\theta(t))\rightarrow 0$ as $t\rightarrow \infty$ with $y(t) \leq 0$, and the forward velocity input $v(t)$ remains nonnegative for all $t\geq 0$.
\end{theorem}

\begin{proof}
Consider again the BAR-FLi Lyapunov function introduced in~\cite[Thm.~8]{Part1_todorovskiCLF2025},
\begin{equation}\label{eq:BAR-Fli_lyap}
V_{\delta \gamma}(\delta,\gamma) = 4 \tan^2 \frac{\delta}{2} + q^2 z^2\,,
\end{equation}
where $z$ is defined in~\eqref{eq:z_BARFLi} and $q = \sqrt{k_1/k_3}$. Differentiating \eqref{eq:z_BARFLi} in closed-loop with \eqref{eq-v-general} and \eqref{eq-omega-general} gives
\begin{align}
\dot z = -\tilde\omega + \frac{k_1 k_2 \sin(2\gamma)}{2\cos^2(\delta/2) 
+ 32 k_2^2 \sin^2(\delta/2)}\,,
\end{align}
while \eqref{eq:z_BARFLi} implies $2z - 2\gamma = 
\arctan\left(4k_2\tan\frac{\delta}{2}\right)$, so that \eqref{eq:N_BARFLi} 
and \eqref{eq:psi_definition} yield
\begin{align}
\frac{k_1}{2}\sin(2\gamma) = k_1 z\,\psi(z,\gamma) 
- \frac{2k_1k_2\tan(\delta/2)}{N(\delta)}\,.
\end{align}
The time derivative of~\eqref{eq:BAR-Fli_lyap} then yields
\begin{align}
\dot{V}_{\delta \gamma} =& -\frac{8k_1k_2\tan^2(\delta/2)}{N(\delta)}\left(1+\tan^2\displaystyle\frac{\delta}{2}\right)\nonumber\\
&+ 2q^2\,z\biggl[\frac{k_1k_2 \sin(2\gamma)}{2\cos^2(\delta/2)+32 k_2^{2}\sin^{2}(\delta/2) }\nonumber\\
&+ 2k_3 \psi(z,\gamma)\left(1 + \tan^2\frac{\delta}{2}\right)\tan\frac{\delta}{2} -\tilde{\omega} \biggr] .
\end{align}
Noting that $N(\delta) >0$, from Young's inequality, we get
\begin{align}
&4k_1z\psi(z,\gamma)\tan(\delta/2) \nonumber\\
&\qquad \leq  \frac{2k_1k_2\tan^2(\delta/2)}{N(\delta)} + \frac{2k_1}{k_2}N(\delta)\psi^2(z,\gamma)z^2
\,.
\end{align}
Then, 
\begin{align}
\dot{V}_{\delta\gamma} \leq& -\frac{6k_1k_2\tan^2(\delta/2)}{N(\delta)}\left(1+\tan^2\displaystyle\frac{\delta}{2}\right)\nonumber\\
&+ \frac{2k_1}{k_2}N(\delta)\psi^2(z,\gamma)\left(1+\tan^2\displaystyle\frac{\delta}{2}\right)z^2 \nonumber\\
&+2q^2z\Biggl[\frac{k_1k_2\sin(2\gamma)}{2\cos^2(\delta/2)+32k_2^{2}\sin^{2}(\delta/2)}-\tilde\omega\Biggr]
\,.
\end{align}
Choosing the $\tilde\omega$ as \eqref{eq:omega_nonovershoot} we get
\begin{equation}\label{eq:Vdot_ndf_nonovershoot}
\dot{V}_{\delta\gamma} \leq -\frac{6k_1k_2}{N(\delta)}\left(1+\tan^2\displaystyle\frac{\delta}{2}\right)\tan^2 \frac{\delta}{2}
-2k_4q^2z^2
\,.
\end{equation}  
Thus, under the control laws \eqref{eq-v-general} and \eqref{eq-omega-general} with \eqref{eq:omega_nonovershoot}, the time derivative of the CLF
\begin{align}
V(\rho,\delta,\gamma) = \rho^2 + 4\tan^2 \frac{\delta}{2} + q^2 z^2\,,
\end{align}
is negative definite on $\{\rho \ge 0\} \times \mathcal{T}_2$. By reasoning similar to that of the proof of~\cite[Thm. 8]{Part1_todorovskiCLF2025}, the point $\rho = \delta = \gamma = 0$ is GAS on $\mathcal{S}_2$, implying that $(x(t), y(t), \theta(t)) \to 0$ as $t \to \infty$.

For the second part, observe that with \eqref{eq-omega-general} and \eqref{eq:omega_nonovershoot} the time derivative of the backstepping transformation \eqref{eq:z_BARFLi} yields
\begin{align}\label{eq-zdot-nonovershoot-tangent2}
\dot{z} = -\left(k_4 + \dfrac{k_3}{k_2}
{\psi^2(z,\gamma)}
N(\delta)\left(1+\tan^2\displaystyle\frac{\delta}{2}\right)
\right)z
\,,
\end{align}
which guarantees $z(t) \geq 0$ and is decreasing for all $z_0 \geq 0$. Then, if $\delta_0 \in (0,\pi)$, $\gamma_0 \in (-\pi/4,\pi/4)$, and $k_2$ is chosen as in \eqref{eq:nonovershoot_k2_interval}, the expression \eqref{eq:z_BARFLi} guarantees that $z_0 \in (0,\pi/4)$ and hence $z(t) \in [0,\pi/4)$ for all $t \geq 0$. Given that, therefore,  $\frac{k_1\cos(2z(t)) }{2 N(\delta(t))}\tan(2z(t))\geq 0$ and rewriting $\dot{\delta}$ as
\begin{align}
\dot{\delta} &=  \frac{k_1\cos(2z) }{2 N(\delta)}
\left(\tan(2z)-4k_2\tan\frac{\delta}{2}\right)\label{eq:delta_dot_nonovershoot}
\,,
\end{align}
we observe that $\dot{\delta} \geq -\frac{2k_1k_2}{N(\delta)}\tan(\delta/2)$ which for $\delta = 0$ implies $\dot{\delta} \geq 0$. Hence, $\delta(t)$ remains nonnegative for all $t \geq 0$. Additionally, from \eqref{eq:delta_dot_nonovershoot}, we observe that whenever $\delta(t) \geq 2\arctan\left(\frac{\tan(2z(t))}{4k_2}\right)$ for any $t \geq 0$, it follows that $\dot{\delta}(t) \leq 0$ and since $z(t) < \pi/4$ for all $t \geq 0$, $\delta(t) < 2\arctan(\infty) = \pi$ holds and ensures that $\delta(t) < \pi$ for all $t \geq 0$ if $\delta_0 \in (0,\pi)$. Thus, $\delta(t) \in [0, \pi)$ for all $t \geq 0$, which implies that $y(t) \leq 0$ for all $t \geq 0$. Finally, since $z(t) \in [0, \pi/4)$ and $\delta(t) \in [0, \pi)$ for all $t \geq 0$, then by \eqref{eq:z_BARFLi} it follows that $\gamma(t) \in (-\pi/4, \pi/4)$, yielding $\cos(\gamma(t)) > \sqrt{2}/2$ and thus guarantees that the forward velocity control law \eqref{eq-v-general} remains nonnegative for all $t \geq 0$.
\end{proof}

\begin{figure}[t]
\centering
\begin{subfigure}{0.45\textwidth}
\centering
\includegraphics[width=0.5\textwidth]{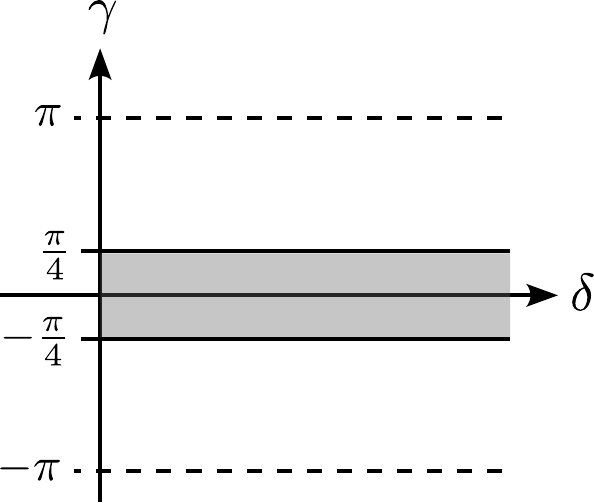}
\caption{Admissible set in $(\gamma,\delta)$}
\label{fig:nonovershoot1}
\end{subfigure}
\vspace{0.5cm}
\begin{subfigure}{0.45\textwidth}
\centering
\includegraphics[width=0.6\textwidth]{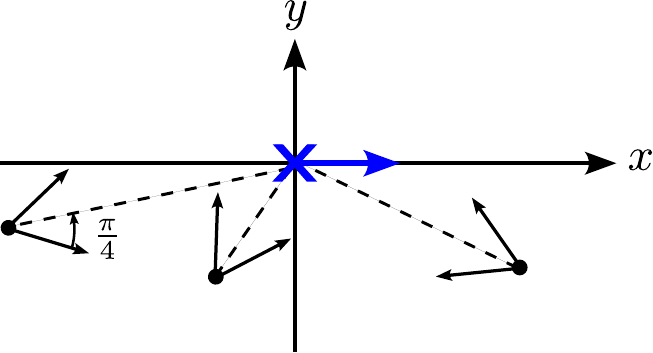}
\caption{Geometry of the LOS angle $\gamma$}
\label{fig:nonovershoot2}
\vspace{-0.5cm}
\end{subfigure}


\caption{Angular variables in the nonovershooting controller design. (a) Admissible set in $\delta$ and $\gamma$ for guaranteed nonovershoot. (b) Admissible LOS angle $\gamma$ for selected initial positions relative to the target's position and heading (blue).}
\vspace{-0.4cm}
\label{fig:nonovershoot_fig}
\end{figure}

Fig.~\ref{fig:trajectory_nonovershoot} presents simulations illustrating enforcement of the curb constraint under a range of initial conditions. The strict enforcement of the safety constraint is verified in Fig.~\ref{fig:states_nonovershoot}, where $\delta(t) \geq 0$ holds for all initial conditions and time, ensuring the system remains in the lower half-plane. Notably, for $\delta_0 \approx \pi$ (grey trajectory), the controller suppresses the growth of $\delta(t)$ to prevent violation.


For several initial conditions (blue, cyan, and lavender), the trajectory turns abruptly once the unicycle reaches the boundary; a sharp jump in $\gamma(t)$ in Fig.~\ref{fig:states_nonovershoot}, akin to the discontinuity typical of conventional safety filters. The corresponding spikes in $\omega(t)$ (Fig.~\ref{fig:control_nonovershoot}) stem not from controller discontinuity but from the large $k_2$ values required ($392.85$, $221.09$, and $127.65$, respectively). In $\tilde{\omega}$ of \eqref{eq:omega_nonovershoot}, large $k_2$ keeps the last term small for large $\delta$, since the $k_2^2$ term in the denominator dominates the $k_2$ term in the numerator. As $\delta \to 0$ the denominator shrinks and the term grows rapidly, then vanishes once $\gamma(t)$ reaches zero at the boundary. The input is therefore impulse-like but nonetheless continuous.


\begin{figure}[t!]
\centering
\includegraphics[width=.9\linewidth]{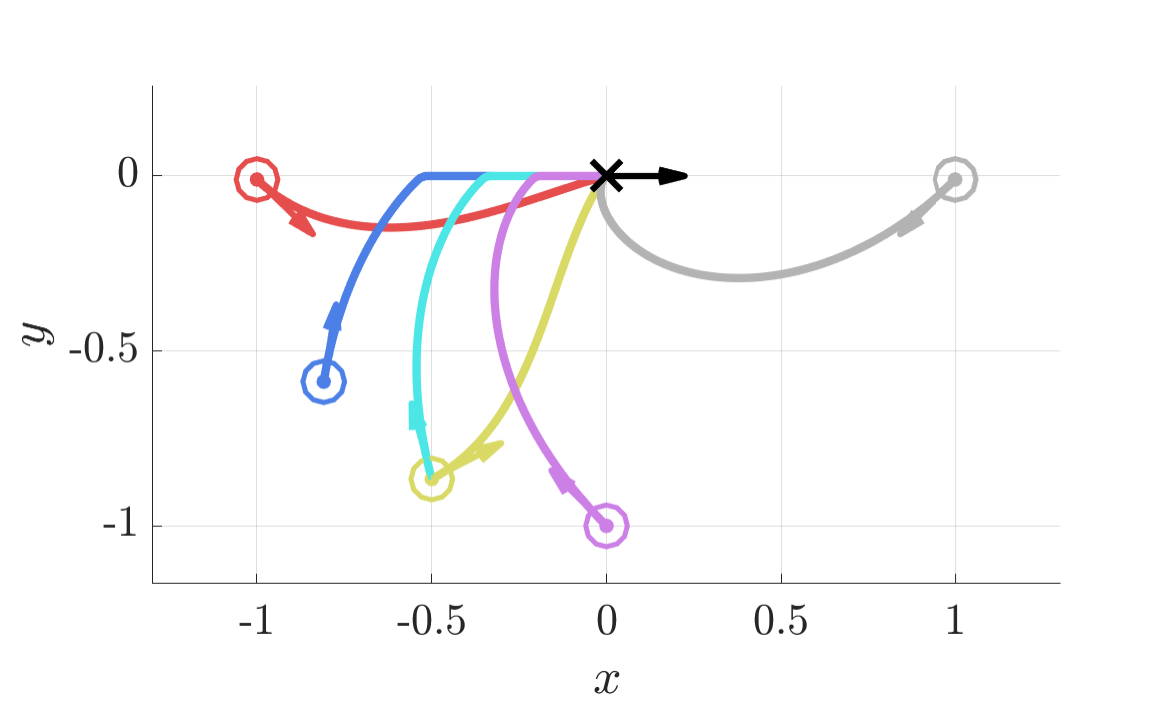}
\caption{Trajectories under the nonovershooting controls \eqref{eq-v-general} and \eqref{eq-omega-general}, using \eqref{eq:omega_nonovershoot} with gains $k_1 = k_3 = k_4 = 1$, and $k_2$ chosen within the interval in \eqref{eq:nonovershoot_k2_interval}. Different colors correspond to different initial conditions.}
\label{fig:trajectory_nonovershoot}
\end{figure}

\begin{figure}[t!]
\centering
\begin{subfigure}{\linewidth}
\centering
\includegraphics[width=.9\linewidth]{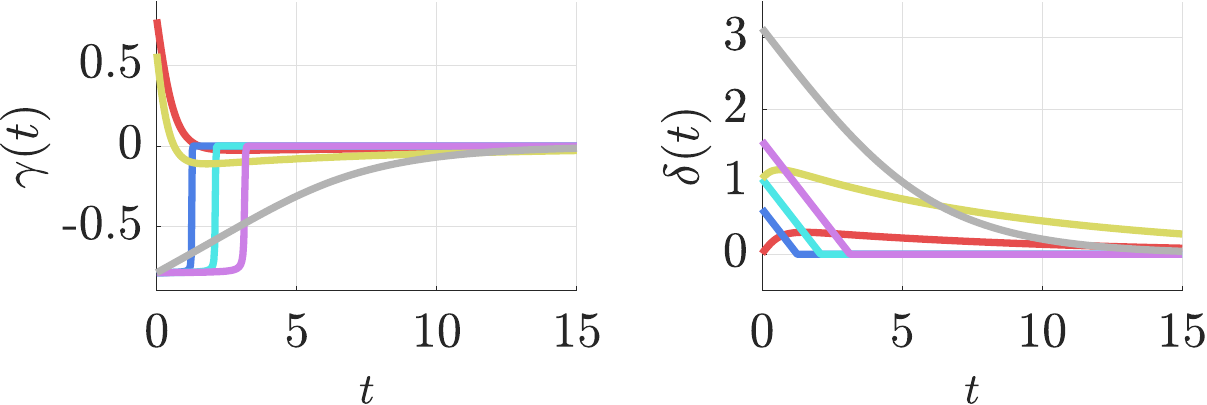}

\caption{$\delta(t)$ remains nonnegative for all time and $\gamma(t)$ turns abruptly in some trajectories (blue, cyan, lavender) once $\delta(t)$ reaches zero.}
\label{fig:states_nonovershoot}
\end{subfigure}
\vskip\baselineskip
\begin{subfigure}{\linewidth}
\centering
\includegraphics[width=.9\linewidth]{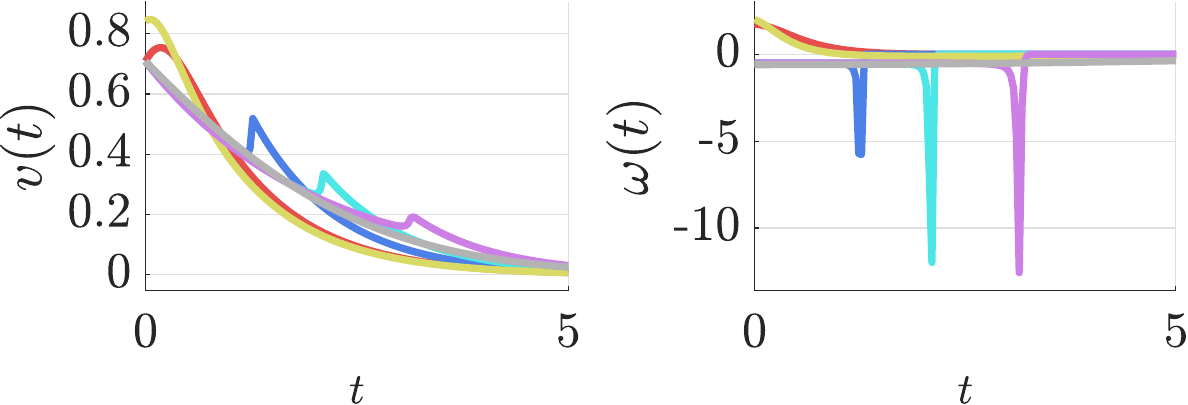}
\caption{Control inputs $v$ and $\omega$ for the nonovershooting trajectories.}
\label{fig:control_nonovershoot}
\end{subfigure}

\caption{Simulation results under nonovershooting control. Colors match the initial conditions in Fig.~\ref{fig:trajectory_nonovershoot}.}
\vspace{-0.4cm}
\label{fig:nonovershoot_combined}
\end{figure}

\section{Conclusion}

This paper presents constructive analytical results for nonmodular exponential stabilization, inverse optimality, robustness, adaptive control, prescribed/fixed-time convergence, and safety of the unicycle system. We develop nonmodular backstepping feedback laws that guarantee global exponential stability with either unidirectional or bidirectional velocity actuation. Leveraging families of strict CLFs, we introduce a general inverse optimal framework that yields broad families of stabilizers with infinite gain margins, ensuring robustness under arbitrarily low input saturation and actuator uncertainty. To further address uncertainty, we design adaptive laws for unknown input coefficients that achieve improved transient performance with reduced initial control effort. We further present controllers that ensure stabilization within a user-defined time horizon through two systematic approaches: a time-dilated prescribed-time design that guarantees smooth convergence of both states and inputs to zero, and a homogeneity-based fixed-time control scheme that offers an analytical bound on the settling time. Finally, nonovershooting control ensures geometric safety constraints and strictly forward motion without optimization-based enforcement. Together with Part I, this work establishes a broad constructive foundation for the stabilization of the nonholonomic unicycle, unifying multiple advanced control properties in a systematic and analytically transparent framework.



\section*{References}
\vspace{-0.5cm}
\bibliographystyle{IEEEtranS}
\bibliography{root}

\vspace{-1 cm}

\begin{IEEEbiography}
[{\includegraphics[width=1in,height=1.25in,clip,keepaspectratio]{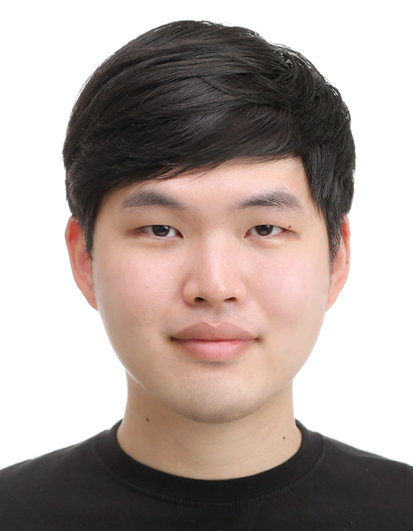}}]{Kwang Hak Kim}
received the B.S. degree in Aerospace Engineering from the University of Texas at Austin, Austin, TX, USA, in 2020, and the M.S. degree in Mechanical Engineering from the University of California San Diego, La Jolla, CA, USA, in 2024. He is currently pursuing the Ph.D. degree in Mechanical and Aerospace Engineering at the University of California San Diego. His research interests include nonlinear safety-critical control of nonholonomic systems.
\end{IEEEbiography}

\vspace{-1 cm}

\begin{IEEEbiography}[{\includegraphics[width=1in,height=1.25in,clip,keepaspectratio]{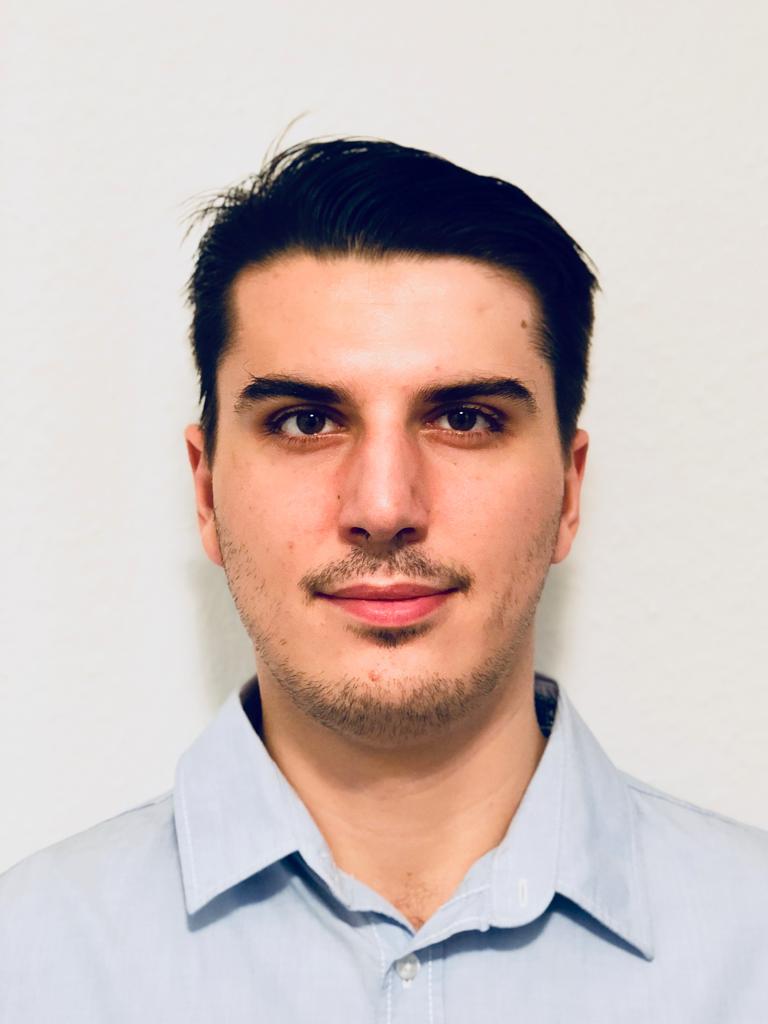}}]{Velimir Todorovski}

received the B.Sc. degree in Electrical Engineering, Electronics and Information Technology from the University of Erlangen-Nuremberg, Erlangen, Germany, and the M.Sc. degree in Electrical and Computer Engineering from Technical University of Munich, Munich, Germany, in 2019 and 2022, respectively.
He is currently a Ph.D. student with the Department of Mechanical and Aerospace Engineering,  University of California San Diego, La Jolla, CA, USA. His research interests include model-free optimization and safe stabilization of nonholonomic systems.
\end{IEEEbiography}

\vspace{-1cm}
\begin{IEEEbiographynophoto}{Miroslav Krsti\'{c}}
(Fellow, IEEE) For the photograph and biography of Miroslav Krsti\'{c}, see the January, 2026 issue of the IEEE Transactions on Automatic Control.
\end{IEEEbiographynophoto}

\end{document}